\documentclass[11pt,a4paper]{article}

\usepackage[disable]{todonotes}
\usepackage{amsmath,amssymb}
\usepackage{amsthm}
\usepackage{graphicx}
\usepackage{tikz}
\usepackage{enumitem}
\usepackage{complexity}
\usepackage{mathrsfs}
\usepackage{comment}
\usepackage{hyperref}
\usepackage[ruled]{algorithm}
\usepackage{algpseudocodex}
\theoremstyle{plain}
\newtheorem{theorem}{Theorem}
\newtheorem{lemma}[theorem]{Lemma}
\newtheorem{proposition}[theorem]{Proposition}
\newtheorem{fact}[theorem]{Fact}
\newtheorem{claim}[theorem]{Claim}
\newtheorem{definition}[theorem]{Definition}
\theoremstyle{remark}
\newtheorem*{claimproof}{Proof of claim}
\theoremstyle{plain}
\newtheorem{problem}{Problem}

\newlist{alphaenumerate}{enumerate}{1}
\setlist[alphaenumerate,1]{label=\alph*.,ref=\alph*,leftmargin=2em}

\let\nolinenumbers\relax % was provided by the LIPIcs class

\newcommand{\MINSP}{Min Pebbling}
\newcommand{\MINWP}{Min Saturated Pebbling}

\newcommand{\inv}{^{\raisebox{.2ex}{$\scriptscriptstyle-1$}}}

\DeclareMathOperator{\out}{out}
\DeclareMathOperator{\inn}{in}
\DeclareMathOperator{\col}{col}

\DeclareMathOperator{\be}{beg}
\DeclareMathOperator{\en}{end}

\DeclareMathOperator{\rank}{rank}
\DeclareMathOperator{\select}{select}

\newcommand{\mt}{\mathcal}
\newcommand{\ms}{\mathscr}
\newcommand{\mf}{\mathfrak}
\newcommand{\mr}{\mathrm}

\DeclareMathOperator{\ESUCC}{succ}

\DeclareMathOperator{\EPREV}{prev}

\DeclareMathOperator{\dd}{.\!.}

\newcommand{\pb}[1]{{\color{blue!70!black}\textbf{PB:} #1} } 
\newcommand{\bri}[1]{{\color{orange!70!black}\textbf{BRI:} #1} } 
\newcommand{\gdv}[1]{{\color{yellow!70!black}\textbf{GDV:} #1} } 
\newcommand{\gyn}[1]{{\color{green!70!black}\textbf{Y.G.:} #1} } 
\renewcommand{\pb}[1]{} 
\renewcommand{\bri}[1]{} 
\renewcommand{\gdv}[1]{} 
\renewcommand{\gyn}[1]{} 

\newcommand{\edgeRank}{\mt R}
\newcommand{\guardianCollection}{\ms G}
\newcommand{\componentSize}{\kappa}
\newcommand{\cutNum}{\partial}
\newcommand{\pathSet}{\mt H}

\title{Compact Path Representation in DAGs via Colored Edge Pebbling}

\author{Paola Bonizzoni$^{1}$ \and Alessio Conte$^{2}$ \and Gianluca Della Vedova$^{1}$ \and Younan Gao$^{1}$ \and Roberto Grossi$^{2}$ \and Brian Riccardi$^{1}$}

\date{%
$^{1}$Department of Computer Science, University of Milano-Bicocca, Italy\\
$^{2}$Department of Computer Science, University of Pisa, Italy\\
\emph{Emails:} paola.bonizzoni@unimib.it, alessio.conte@unipi.it, gianluca.dellavedova@unimib.it, younan.gao@unimib.it, roberto.grossi@unipi.it, brian.riccardi@unimib.it}

\begin{document}

\maketitle

\begin{abstract}
Compactly representing a variation graph is a core problem in computational pangenomics that is usually attacked with techniques  that have been originated on texts and  adapted to graphs. In this paper we propose a new framework that takes a topology-centric perspective instead. A variation  graph is modeled as a directed acyclic graph (DAG) together with a set of distinguished paths, where each path is assigned a distinct color. Our compact representation is centered on \emph{pebbling} the graph, i.e. placing colored pebbles on edges so that every predefined path can be univocally reconstructed from the pebbled edges. In particular, a saturated pebbling marks each chosen edge with every path
(color) traversing it.

We first propose a data structure to represent and query a variation graph  with  storage space depending  on the size of the pebbling.
The supported queries are: (i) \emph{path query}, which recovers a path given its
color, and (ii) \emph{edge query}, which reports the colors of paths traversing
a given edge. 
We then prove that the problem of finding a pebbling of minimum size is solvable 
in polynomial time. On the contrary, we prove that finding a \emph{saturated} 
pebbling of minimum size is \NP-hard, but can be reduced to the minimum-weight set cover problem, allowing us to leverage integer linear programming (ILP) solvers. We show how to exploit saturated pebblings to achieve faster queries times than minimum size pebbling. 

Our framework opens a new algorithmic viewpoint on developing more efficient variation graph representations rooted on the study of the topology of those graphs.
%\bri{TODO: decide which other contributions to add to abstract and paper (e.g. ILP).}
\end{abstract}

\newpage
\section{Introduction}
Computational pangenomics is an emerging field in algorithmic bioinformatics that addresses the limitations of a single linear reference genome by moving towards graph-based representations capable of capturing genomic variation across entire populations~\cite{baaijens2022computational,computational2018computational,liao2023draft}. 
In these models, genome sequences are represented as paths in a graph whose vertices are labeled with subsequences~\cite{baaijens2022computational}. 
Such \emph{pangenome graphs} allow many genomes to be represented within a single structure, but they also introduce algorithmic challenges in representing, indexing, and querying these graphs efficiently. 

A central goal is to obtain compact representations of a collection of genomic sequences that still support useful query operations. In particular, a pangenome graph exhibits   local topological structures, known as \emph{bubbles}  \cite{onodera2013detecting},   where multiple distinct sequence paths diverge from a common starting node and  converge at a common ending node, representing alternative genetic variants. In colored pangenome graphs, called  variation graphs,   each genome sequence is distinguished by assigning a color to the path identifying the sequence. Then, the colored  bubbles in variation graphs capture at a topological level the main variations that differentiate one genomic sequence from another. Capturing bubbles that distinguish a colored path from another is particularly relevant in applications such as sequence to graph alignment to pangenome graphs —where the goal is to efficiently locate the subpath that best matches an input sequence— or path inference from a collection of strings, which involves identifying a path whose sequence-labeled vertices match most strings in the collection \cite{Dentiet-al2026}. \gyn{\cite{Dentiet-al2026} is claimed to be published on WABI2026, but note that its notification day is in August.} 
Although the topological structure of the graph is relevant in these applications, most existing work focuses on indexing genome sequences for graph querying and identifying sequence-labeled paths, whereas comparatively little attention has been devoted to exploiting the graph's topology. Given this graph representation, a natural  question is how to compactly represent   the specific  topology of the graph, combined with the sequence-labeled vertices, to uniquely identify and differentiate the paths underlying the genome sequences. 

Motivated by  this question from a theoretical perspective, in this work, we abstract away sequence information and study compact graph representations based on the graph's topology that can be used to identify edges that differentiate a predefined set of paths. In doing so,   we   use   edge colors of bubbles in the graph to unequivocally identify each colored path in the graph.
Our main objective is to explore  novel algorithmic approaches to indexing and querying variation graphs that  complement  path labeling with the topological structure, therefore leveraging the later for graph-based applications in pangenomics.
Under this aim, we 
formalize  the notion of \emph{colored edge pebbles}. 
Intuitively, a pebble of color $i$ placed on an edge indicates that the edge belongs to the $i$-th predefined path. Then we investigate the following general question on colored graphs: How to place colored pebbles only on few edges of predefined paths while still enabling unambiguosly the full recovery of the path? 
We address this question by  studying the  following computational problem: given a DAG $G$, which we call variation graph,  defined as the triple $(V, E, \mathcal{H})$, where $\mathcal{H}$ is a set of paths such that each edge in $G$ is traversed by at least one path,  each identified by an integer (or \emph{color}) in $\{1, \cdots,|\mathcal{H}|\}$, we want to find  the minimum number of edges in each path $h_i \in \mathcal{H}$ that must be colored with $i$ so that the entire path of that color can be reconstructed without ambiguities. While path reconstruction may follow different criteria, in this paper we adopt a natural one: pebbling under the \emph{unique path} condition, that is, placing pebbles so that each pebbled edge can be reached from the previously pebbled edge of the same color through a unique path. 

%Instead of labeling every edge of every path, we place pebbles on only a subset of edges while still allowing each predefined path to be reconstructed unambiguously. 
Figure~\ref{fig:guardian-lists} illustrates the idea.

\begin{figure}[!h]
    \centering
    \includegraphics[scale=1]{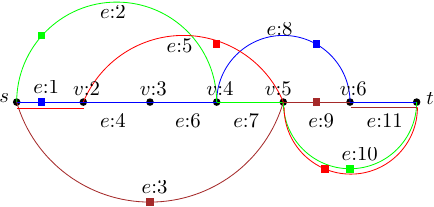}
    \caption{Example of pebbling and guardians.
        The set $\mt H$ consists of four paths (sequence of edges): $h_{\mr{green}}=(e_2,e_7,e_{10})$, $h_{\mr{red}}=(e_1,e_5,e_{10})$, 
        $h_{\mr{blue}}=(e_1, e_4, e_6, e_8, e_{11})$, and $h_{\mr{brown}}=(e_3, e_9, e_{11})$.
        The elements of the chosen pebbling $\mt P$ are pictured as small colored squares over the edges. For example, 
        $(e_5, \mr{red})$, $(e_{10}, \mr{green})$ and $(e_8, \mr{blue})$ are some of the elements of $\mt P$. From $\mt P$ we derive the following
        guardians, one for each path: $\mt G_\mr{green}=(e_2,e_{10})$, $\mt G_\mr{red}=(e_5,e_{10})$, $\mt G_\mr{blue}=(e_1,e_8)$, 
        and $\mt G_\mr{brown}=(e_3,e_9)$.  
    }\label{fig:guardian-lists}
        % $\mt P = \left\{(e_2, 1), (e_7,1), (e_{10},1), (e_1, 2), (e_5, 2), (e_{10}, 2), (e_1, 3), (e_4, 3), (e_6, 3), (e_8, 3), (e_{11}, 3)
        % (e_3, 4), (e_9, 4), (e_{11}, 4)\right\}$
\end{figure}

%The graphs considered in this paper are directed acyclic graphs (DAGs), defined as $G=(V,E,\mt H)$, where $\mt H$ is a set of paths in $G$ such that every edge of $G$ is traversed by at least one path. 
%This assumption is natural in variation graphs, where each edge represents a genomic segment observed in at least one haplotype. 
%Each path is identified by an integer (or \emph{color}) in $\{\,1,\ldots,|\mt H|\,\}$. 
We refer to such graphs as \emph{variation graphs}, and to the process of assigning colored pebbles to selected edges as \emph{pebbling}. 

To ensure that paths can be reconstructed from the pebbled edges, we introduce a structural constraint called the \emph{unique path condition}. 
Intuitively, pebbles are placed so that each pebbled edge of a given color can be reached from the previously pebbled edge of the same color via a unique path in the graph, enabling unambiguous reconstruction of the path. We observe that the study of pebbling problem in digraphs is of more general theoretical interest and it is not limited to pangenomics.  However,   it opens new perspectives that may impact the field of graph pangenomics with practical  topological-oriented  approaches to address  main fundamental questions  on variation-graphs, such as the path-inference problem that asks for distinguishing paths in the graph. While this problem is currently solved at the sequence level by mainly considering  sequence-labeled vertices, the compact graph representation based on pebbling that we propose in the paper,  combined with  sequence-labeling of paths,  could provide   alternative algorithmic solutions to path detection.

\paragraph*{Our contributions.}

We introduce two new computational problems defined on a variation graph $G=(V,E,\mt H)$. 
For a path $h_i \in \mt H$, we call a sequence of edges satisfying the unique path condition a \emph{guardian} for $h_i$. 
In the \MINSP{} problem, the goal is to compute a pebbling for $G$ of minimum cardinality (size). Another natural variant of \MINSP{} arises when we assume that whenever an edge is pebbled by some color, then it is pebbled with every color of a path traversing it. We call this property \emph{saturation} of an edge and \emph{saturated} pebbling assume each edge to satisfy the property. Intuitively, we add more information on pebbled edges and thus we expect to simplify the reconstruction and queries step of paths.  Thus, in the \MINWP{} problem, the goal is to construct a minimum-size saturated pebbling.

The two problems differ fundamentally in how they treat overlaps between paths. 
The \MINSP{} problem follows a \emph{local} optimization principle: each path can be processed independently to obtain a pebbling of minimum size. 
In contrast, \MINWP{} follows a \emph{global} optimization principle: if an edge is selected for pebbling, it is effectively pebbled for every path that traverses it. In other words, the new optimization criteria 
induces dependencies between paths,  while in the \MINSP{} we can assume a  path-wise decomposition.

For \MINSP{}, we present two algorithms with running times $\mt O(|\mt H|\cdot|E|)$ and $\mt O(|E|\cdot|V| + \sum_i |h_i|)$, respectively. 
In contrast, we prove that \MINWP{} is \NP-hard but can be reduced to the minimum-weight set cover problem, allowing us to leverage integer 
linear programming (ILP) solvers.
Moreover, our pebbling framework naturally supports two fundamental queries:
\begin{itemize}
\item \emph{Edge query:} given an edge $e$, return the set $\col(e)=\{\, i \in \{\,1, \dots, |\mt H|\,\} : e \in h_i \,\}$.
\item \emph{Path query:} given a color $i$, sequentially return the edges of the path $h_i$.
\end{itemize}

We design compact data structures for representing a pebbling $\mt P$, whose space usage
depends on the size of $\mt P$, and for supporting edge and path queries. Specifically, the data structure consumes 
$\mt O(|E|+|\mt P|\log|\mt H|)$ additional bits of space on top of $G$, and supports edge queries in $\tilde{O}(|E|+|\pathSet|+|\mt P|)$ 
time\footnote{The notation $\tilde{O}(f(n))$ suppresses polylogarithmic factors.} and path queries in $\mt O(|E|)$ time. In case $\mt P$
is a minimum size pebbling, we improve edge queries to $\tilde{\mt O}(|E|+|\pathSet|)$ time.
We then show that if $\mt P$ is saturated (not necessarily of minimum size), an edge query for $e \in E$ can be answered in $\tilde{\mt O}(|\componentSize(e)| + |\cutNum(e)| + |\mt H|)$ time,
where $\componentSize(e)$ is the set of edges of the connected component of $G$ containing $e$ obtained after removing all pebbled edges, 
and $\cutNum(e)$ are the ``border'' pebbled edges, that is, those pebbled edges incident on vertices of $\componentSize(e)$ from the outside.

\paragraph*{Related work.}
As already observed, the computational problems we formalize in this work are of general theoretical interest, beyond graph pangenomics, which mainly motivated them. Indeed, we recently discovered an independent work \cite{Milani2024} in which a notion similar to guardians has been investigated to define a new parameter for digraphs called ear anonymity. 
Compact topological representations of a pangenome are not common in the literature even though previous works focused on the topology of a variation graph for solving other problems, such as identifying \emph{safe} walks or complex structures in variation-graphs or a minimum path-cover in digraphs \cite{cairo2022safety, caceres2022sparsifying}. Relevant complex structures that are captured by the edges of a guardian list for a colored path, as already stated above are bubbles \cite{onodera2013detecting}. 
%Bubbles represent vertex-disjoint paths between a pair of vertices that are the source and the sink of the bubble. 
% In our work, we formalize the ILP for the \MINWP{} in terms of set-cover for bubbles, showing a connection between the problem of identifying bubbles and colored pebbling. 
In our work, we formulate \MINWP{} as an ILP based on a set-cover representation of bubbles, thereby establishing a connection between bubble identification and colored pebbling.

Compact representations of graphs have been extensively studied in the context of colored de Bruijn graphs~\cite{alanko2023themisto,fan2024fulgor}, where the objective is to construct space-efficient representations of k-mer sets that support matching queries.
Another established line of research focuses on indexing graphs for pattern matching~\cite{equi2023complexity}, including extensions to colored variation graphs~\cite{siren2014indexing}.
A prominent approach for indexing predefined paths in variation graphs is the Graph Burrows--Wheeler Transform (GBWT)~\cite{siren2017indexing,siren2020haplotype}. 

In the GBWT, vertices are treated as characters and each path is interpreted as a string over the alphabet $\{\, 1,\dots,|V|\, \}$, terminated by a distinct end-marker.
The structure is akin to a generalized FM-index~\cite{fm-index-2000} built over the concatenation of the paths. 
Let $r$ denote the number of runs in the Burrows--Wheeler Transform (BWT)~\cite{BurrowsWheeler1994}. 
The GBWT occupies $\mt O(r+|E|+|\mt H|)$ words of space and supports path extraction in $\mt O(|h_i|\cdot t_\mr{LF})$ time, where $t_\mr{LF}$ denotes the time required for LF-mapping queries on the BWT.
%By incorporating the $r$-index \cite{r-index-2020}, the structure can support edge queries in $\mt O(|\col(e)|\cdot t_{LF})$ time, preserving the space complexity.
Edge queries require auxiliary sampling structures, incurring in a space overhead of $\mathcal{O}(\|\mt H\|/d)$ and a query time of $\mathcal{O}(|\col(e)| \cdot d \cdot t_\mr{LF})$. Here, $\|\mt H\|$ denotes the cumulative of all paths in $\mt H$ and $d$ is a sampling parameter.

While the GBWT provides a (multi-)string-based indexing solution, our approach shifts from this sequence-centric indexing paradigm towards a topology-centric representation. By exploiting the \emph{unique path condition}, we identify the minimal structural information required to reconstruct predefined paths. 
Leveraging this condition as a guiding principle for compression constitutes the main conceptual contribution of this work.

\section{Preliminaries}
\label{sect-prelim}
We work in the word-RAM model with word size
$w=\Omega(\log N)$, where $N$ denotes the size of the input instance;
thus, vertex identifiers, edge ranks, color identifiers, and memory
addresses fit in $O(1)$ words, and standard operations on words take
$O(1)$ time.

\paragraph*{Rank and select queries.}  
Given an array $A[1 \dd n]$ over the alphabet $\{\,1, \dots, \sigma\,\}$, the operation $\rank_c(A, j)$ returns the 
number of occurrences of $c$ in $A[1 \dd j]$, for $c \in \{\,1, \dots, \sigma\,\}$.
The operation $\select_c(A, j)$ returns the position of the $j$-th occurrence of $c$ in $A$, if it exists, 
and returns $n+1$ otherwise (\emph{i.e.}, when $A$ contains less than $j$ occurrences of $c$).

\begin{lemma}[{\cite[Theorem 2.2]{golynski2006rank}}]\label{lem-rank-select}
%Given an array $A[1 \dd n]$,\gdv{revised statement, please check} there exists a data structure using $n \log \sigma + o(n \log \sigma)$ bits  that supports $\rank$ queries in $\mt O(\log \log \sigma)$ time and $\select$ queries in $\mt O(1)$ time.
For an array $A[1 \dots n]$ over an alphabet of size $\sigma$, there exists a 
data structure using $n \log \sigma + o(n \log \sigma)$ bits of space that supports 
$\rank$ queries in $O(\log \log \sigma)$ time and $\select$ queries in $O(1)$ time.
\end{lemma}

% \paragraph*{Integer sorting.}  
% We remind that in the word RAM model, integer sorting can be performed in $o(n \log n)$ time.

% \begin{lemma}[{\cite[Theorem 1]{han2004deterministic}}]\label{lem-sorting}
% Let $A[1 \dd n]$ be an array of $n$ integers from $\{1, \dots, U\}$.  
% The integers of $A$ can be sorted in $\mt O(n \log \log n)$ time.
% \end{lemma}

\paragraph*{Graphs.}
We assume basic familiarity with graph theory. Let $G = (V, E)$ be a directed acyclic graph (DAG).
We assume that the vertices $V = \{\, 1, 2, \dots, |V| \,\}$ are named with their rank in an arbitrary topological
order (so, $1$ is the first, $2$ the second and so on). This order induces an order of the edges (the lexicographical
order of pairs) and associates to every edge its rank. We denote by $\edgeRank : E \to \{\, 1, \dots, |E| \,\}$ the bijection that
associates to every edge its rank. Whenever we talk about a sequence of vertices (or edges) 
we will assume that such sequence is sorted by their ranks.

% We denote by $\out(v)$ (resp. $\inn(v)$) the sequence of edges outgoing from (resp. incoming to) $v$.
% \gyn{We denote by $\out(v)$ (resp. $\inn(v)$) the sequence of heads of the edges outgoing from $v$ (resp. tails of the edges incoming to $v$).}
We denote by $\out(u)$ the sequence of vertices $v$ such that $(u,v)\in E$, and by $\inn(v)$ the sequence of vertices $u$ such that $(u,v)\in E$, both
sorted by the topological order of $V$.
%\bri{head/tail is not defined.}
%\gyn{\textbf{I searched the keywords "tail" and found that it appears at least 19 times. Please fix this issue.}}\bri{Out of these 19 times, only a few are existing in the manuscript and are all in Section 5. Decide by yourself if you prefer to move this definition there or if you want to keep it here.} 
For an edge $(u,v) \in E$ we call $u$ the \emph{tail} and $v$ the \emph{head}.
A \emph{path} $h$ 
of $G$ is a sequence of consecutive edges. We denote by $\be(h)$ (resp. $\en(h)$) the starting (resp. ending) vertex of $h$, 
and by $V(h)$ the ordered sequence of vertices obtained while traversing $h$.

%In order to represent $G$, $\edgeRank(\cdot)$ and $\edgeRank\inv(\cdot)$, we can use the data structure from Lemma~\ref{lem-rank-select}.
%In particular, we have the following result (whose proof is deferred to the Appendix~\ref{app-graph-rep}).

By augmenting a compact representation of $G$ with the rank-and-select data structure of Lemma~\ref{lem-rank-select}, we can also support $\edgeRank(\cdot)$ and $\edgeRank\inv(\cdot)$. This gives the following proposition. %, whose proof is deferred to Appendix~\ref{app-graph-rep}.

%\gdv{3?}
% \begin{proposition}\label{prop-graph-rep}
%     Let $G=(V,E)$ be a DAG. We extend the data structure in \cite[Section 9.1]{navarro2016compact} that uses $\mt O(|E|\log |V|)$ bits to represents $G$ while supporting listing $\out(v)$ and $\inn(v)$ in $\mt O(|\out(v)|)$ and
%     $\mt O(|\inn(v)|)$ time, computing $|\out(v)|$ and $|\inn(v)|$ in $\mt O(1)$ time, to also return
%     $\edgeRank(e)$ for any $e\in E$  in $\mt O(\log\log |V|)$ time and $\edgeRank\inv(j)$ for any $j\in [1, |E|]$ in $\mt O(1)$ time.
% \end{proposition}

\begin{proposition}\label{prop-graph-rep}
Let $G=(V,E)$ be a weakly connected DAG. There exists a data structure that represents $G$ using $\mt O(|E|\log |V|)$ bits and supports the following operations: for any $v\in V$, it reports $\out(v)$ and $\inn(v)$ in $\mt O(1+|\out(v)|)$ and $\mt O(1+|\inn(v)|)$ time, respectively, and returns $|\out(v)|$ and $|\inn(v)|$ in $\mt O(1)$ time; moreover, it returns $\edgeRank(e)$ in $\mt O(\log\log|V|)$ time and $\edgeRank\inv(j)$ in $\mt O(1)$ time for any $e\in E$
and any $j\in[1 \dd |E|]$.
\end{proposition}

\begin{proof}
The in- and out-neighbor queries can be supported by the data structure described in \cite[Section 9.1]{navarro2016compact}, which occupies $\mt O(|E|\log |V|)$ bits.
It therefore suffices to present solutions for the $\edgeRank$ and $\edgeRank\inv{}$ queries.

To support $\edgeRank(u,v)$, we build the data structure of Lemma~\ref{lem-rank-select} over the array $D[1 \dd |E|]$ that stores the concatenation of 
$\out(1), \out(2), \dots, \out(|V|)$, enabling $\rank$ and $\select$ operations. 
Note that for every edge $(u,v)$ it holds $D[\edgeRank(u,v)] = v$.

We also store an auxiliary array $F[v] = 1 + \sum_{1 \le j < v} |\out(j)|$ for $1 < v \le |V|+1$ and $F[1] = 1$.
The array $F$ occupies $O(|V|\log(|E|+1))$ bits of space.
Since the underlying undirected graph is connected, we have
$|V|\le |E|+1$, and thus this space is
$O(|E|\log |V|)$ bits. 
%\gyn{I added the space cost of the array $F$.}

Then, it follows that $\edgeRank(u,v) = \select_v\!\bigl(D, \rank_v(D, F[u]-1) + 1\bigr)$. By Lemma~\ref{lem-rank-select}, each $\rank$ query takes $\mt O(\log\log |V|)$ time, while each $\select$ query takes $\mt O(1)$ time. Hence, $\edgeRank(u,v)$ can be answered in $\mt O(\log\log |V|)$ time.

To support $\edgeRank\inv(j)$, we store an array $Q[1 \dd |E|]$ such that $Q[j] = (u,v)$ whenever $\edgeRank(u,v) = j$. 
Then, a $\edgeRank\inv(j)$ query can be answered in $\mt O(1)$ time.

The arrays $D[1 \dd |E|]$, $F[1 \dd |V|+1]$, and $Q[1 \dd |E|]$ together occupy $\mt O(|E|\log |V|)$ bits. The additional data structure on $D$ (from Lemma~\ref{lem-rank-select}) also requires $\mt O(|E|\log |V|)$ bits. Thus, the total space usage is $\mt O(|E|\log |V|)$ bits.
\end{proof}

To ease the exposition, from now on we will identify every edge $e \in E$ with its rank $\edgeRank(e)$.

\begin{definition}[Variation Graph]
    A \emph{variation graph} is a triple $G=(V,E,\mt H)$ where $(V,E)$ is a DAG and $\mt H =\{\,h_1,\dots,h_k\,\}$ 
    is a collection of paths of $(V, E)$ such that every edge $e \in E$ is traversed by at least one path in $\mt H$.
\end{definition}

To avoid confusion, we will always specify whether a path is a generic path in the graph (path of $G$) or if it is a selected
path (path in $\mt H$). For every path $h_i \in \mt H$, we let $V_i = V(h_i)$, $s_i = \be(h_i)$ and
$t_i = \en(h_i)$. If $(u, v)$ and $(v, w)$ are consecutive edges of $h_i$,
we call $(v, w)$ the \emph{successor} of $(u, v)$ in $h_i$, and $(u, v)$ the \emph{predecessor} of $(v, w)$ in
$h_i$. If there exists a path in $G$ from $u$ to $v$, we call $u$ an \emph{ancestor of $v$} and $v$ a \emph{descendant of $u$}.
Moreover, if such a path is unique (or if $u=v$) then $u$ is a \emph{cardinal} ancestor and $v$ is a \emph{cardinal} descendant.

An edge $(u,v) \in E$ is a \emph{transitive edge} if there exist at least two paths in $G$ going from $u$ to $v$, one of which is
$(u,v)$ itself. If no edges of $G$ is transitive, $G$ is called \emph{transitive-free}.

For every $e \in E$ and $v \in V$, we let $\col(e) = \{\, i \in \{\,1, \dots, |\mt H|\,\} : e \in h_i \,\}$ and 
$\col(v) = \{\, i \in \{\,1, \dots, |\mt H|\,\} : v \in V_i \,\}$.
For readability, we slightly abuse the notation and write $\col(u,v)$ in place of $\col((u,v))$ for an edge $(u,v)$.

\section{Guardians and Pebbles}
\label{sect-pebbling-intro}

% In this section we introduce the notions of \emph{guardians} and \emph{pebblings} as a way to reduce the 
% space needed for representing a set of paths in a DAG. Specifically, a guardian is a subset of
% the edges of a specific path that allows for its unambiguous reconstruction. A pebbling is an assignment
% of (possibly zero or multiple) colors to each edge in order to determine in which guardian a specific edge will
% be inserted into. If a pebbling can be thought as the raw information needed to reconstruct the paths, a
% collection of guardians is a ``structured'' high-level representation of a pebbling. In the end of this
% section we present also a lightweight data structure to store and query a pebbling.

In this section, we introduce the notions of \emph{guardians} and \emph{pebblings} as a means of reducing the space required to represent a set of paths in a DAG. A guardian is a path-centric notion: it is a subsequence of the edges of a given path that allows the path to be reconstructed unambiguously. A pebbling provides the corresponding graph-centric view by assigning colors to edges, where each color identifies the guardian of a particular path. Thus, a pebbling represents the collection of guardians as a whole.
At the end of the section, we also present a lightweight data structure for storing and querying a pebbling.

\begin{definition}[Unique Path Condition]\label{def:uni-path}
Let $G=(V,E)$ be a DAG, let $v_0,u_{k+1}\in V$, and let
$\mt G=(e_1,\dots,e_k)$ be a sequence of edges of $G$, where
$e_i=(u_i,v_i)$ for every $i\in[1,k]$.
We say that $(v_0,\mt G,u_{k+1})$ satisfies the \emph{unique path condition} if, for every $i\in[0,k]$, there exists a unique directed path in $G$ from $v_i$ to $u_{i+1}$.
\end{definition}

%\gdvi{I think that guardian (without s) is better for the next definition, since the next definition is pebbling (without s)}
\begin{definition}[Guardian]\label{def:guardians}
    Let $G = (V, E, \mt H)$ be a variation graph, and let $h_i \in \mt H$. A \emph{guardian for $h_i$}
    is a subsequence $\mt G_i \subseteq h_i$ such that $(s_i, \mt G_i, t_i)$ 
    satisfies the unique path condition. A \emph{guardian collection for $G$} is a set 
    $\guardianCollection = \{\,\mt G_1, \dots, \mt G_{|\mt H|}\,\}$ such that $\mt G_i$ is a guardian for $h_i$.
    Given an edge $e$, we define $\guardianCollection(e) = \{\, i : e \in \mt G_i \,\}$.
\end{definition}

% Intuitively, a guardian for $h_i$ is a sequence of edges that allows us to unambiguously reconstruct 
% $h_i$ (or $V_i$). 
% We can represent the selected edges in a guardian for $h_i$ with colored pebbles: a pebble of color $i$ is placed on an edge $e$ whenever $e$ 
% belongs to the guardian for $h_i$ (see Figure~\ref{fig:guardian-lists}).

% \textcolor{red}{Younan: Please check which version of pebbling you prefer. The second one does not require to define the universe $\ms P$. }
% \textcolor{blue}{Brian: This is better, we don't need $\ms P$. Thank you!}
\begin{definition}[Pebbling]\label{def:pebbling}
	Let $G = (V, E, \mt H)$ be a variation graph. A subset $\mt P \subseteq \{(e, i): e \in E, i \in \col(e)\}$ is a \emph{pebbling} of $G$ if $\{\, \mt G_1, \dots, \mt G_{|\mt H|}\, \}$ is a guardian collection for $G$, where each $\mt G_i$ is the ordered sequence of edges $( e : (e, i) \in \mt P)$.
	A pebbling $\mt P$ is \emph{saturated} if for every $e \in E$ we have either 
    $\guardianCollection(e) = \emptyset$ or $\guardianCollection(e) = \col(e)$.\end{definition}

By definition, a pebbling induces a guardian collection for $G$. The \emph{saturation} property ensures that
whenever an edge is pebbled, it is pebbled with the colors associated with every path traversing that edge. This will
become useful when considering queries performance. 
%Our goal will be to reduce the size of the pebbling (and, thus, of the induced guardian collection) in order to obtain a lightweight representation of $\pathSet$. 
Our goal will be to combine small-size pebblings (hence small-size induced guardian collections) with efficient algorithms for path and edge queries, to obtain a lightweight and efficient representation of $\pathSet$.
This leads to the following natural optimization problems:

\begin{problem}[\MINSP{} and \MINWP{}]\label{problem:pebbling}
Given $G$, \MINSP{} asks for a minimum-size pebbling $\mt P^*$ of $G$, whereas \MINWP{} asks for a minimum-size \emph{saturated} pebbling $\mt P_s^*$ of $G$.
\end{problem}

% \begin{problem}[\MINSP{}]\label{problem:local-pebbling}
%     Given $G$, compute a pebbling $\mt P^*$ of $G$ of minimum size.
% \end{problem}

% \begin{problem}[\MINWP{}]\label{problem:global-pebbling}
%     Given $G$, compute a \emph{saturated} pebbling $\mt P_s^*$ of $G$ of minimum size.
% \end{problem}

Note that a non-saturating pebbling is of minimal cardinality if and only if each of its induced guardians is individually minimized. Hence, if we do not require saturation, \MINSP{} can be obtained by independently minimizing the size of each (induced) $\mt G_i$.
Figure~\ref{fig:guardian-lists} shows a \MINSP{} for
the variation graph, consisting of eight pebbles in total.
Observe that instead minimum-size saturated pebbling minimizes pebbled edges of the whole graph needed to reconstruct all selected paths of the graph. Thus  \MINWP{} induces a guardian  $\mt G_i$ for each path that is not necessarily of minimum size, but it minimizes the summation of pebbles used compared to those used in  \MINSP{}.

Figure~\ref{fig:saturated guardian-lists} shows a \MINWP{} for the same variation graph as in Figure~\ref{fig:guardian-lists}.
It consists of nine pebbles, one more than the minimum-size pebbling shown in
Figure~\ref{fig:guardian-lists}.

\begin{figure}[!t]
    \centering
    \includegraphics[scale=1]{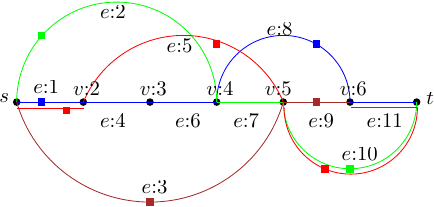}
    \caption{A \MINWP{} for the same variation graph as in Figure~\ref{fig:guardian-lists}. The pebbling consists of nine pebbles, shown as colored squares above the corresponding edges.         
    }\label{fig:saturated guardian-lists}
        % $\mt P = \left\{(e_2, 1), (e_7,1), (e_{10},1), (e_1, 2), (e_5, 2), (e_{10}, 2), (e_1, 3), (e_4, 3), (e_6, 3), (e_8, 3), (e_{11}, 3)
        % (e_3, 4), (e_9, 4), (e_{11}, 4)\right\}$
\end{figure}

\subsection{Data Structure for Representing Pebbled Variation Graphs}

Before going further, we present a space-efficient data structure for representing a guardian collection, alongside
five fundamental operations that form the interface for navigating and querying a pebbled variation graph.

%\gyn{Maybe, we could move the following description of the data structure back to the proof of Theorem, since this paragraph appearing before "Definition" might look strange. The reason that we have it before the theorem in the previous version is we used to defer the whole proof into the appendix but we still wanted to give a high-level idea of the data structure to the reviewers who might not see the appendix. Since we don't have appendix at all, I would suggest to move it back to the proof.}
% \bri{Fine for me!}
% Let $\mt P$ be a pebbling of $G$ and let $\guardianCollection$ be its induced guardian collection. Our data structure for representing $\mt P$  consists of:
% \begin{enumerate}
%     \item An array $P[1 \dd |\mt P|]$ for storing the pebble colors of each edge,
%     \item A bitvector $B[1 \dd |E|+|\mt P|+1]$ containing $|\mt P|$ ones 
%         and $|E|+1$ zeros, and
%     \item Structures supporting $\rank(\cdot)$ and $\select(\cdot)$ over $P$ and $B$ 
%         (see Lemma~\ref{lem-rank-select}).
% \end{enumerate}

% Specifically, the array $P$ consists of the concatenation of $\guardianCollection(1), \dots, \guardianCollection(|E|)$ where every block $\guardianCollection(i)$ is
% sorted increasingly. The bitvector $B$ encodes pebble multiplicities: the number of ones between the $j$-th and 
% $(j+1)$-th zero equals the number of pebbles on edge $j$. Overall, this data structure consumes
% $\mt O(|E| + |\mt P| \log |\mt H|)$ bits of space.

%\gyn{The statement of Lemma~\ref{lem-ds-pebbles} has been restated.}
\begin{lemma}\label{lem-ds-pebbles}
Let $G=(V, E, \mt H)$ be a variation graph and $\mt P$ be a pebbling for $G$.
In addition to a representation of the underlying DAG $(V, E)$ given by Proposition~\ref{prop-graph-rep},  there exists a data structure that uses $\mt O(|E| + |\mt P| \log |\mt H|)$ bits of extra space and supports the following queries:
\begin{alphaenumerate}
    \item For $e\in E$, determine whether $\guardianCollection(e) \neq \emptyset$ in $\mt O(\log \log |V|)$ time,
    \item For $e\in E$, return $\guardianCollection(e)$ in increasing order in $\mt O(\log \log |V| + |\guardianCollection(e)|)$ time,
    \item For $i\in [1 \dd |\mt H|]$, return $|\mt G_i|$ in $\mt O(\log \log |\mt H|)$ time,
    \item For $i\in [1 \dd |\mt H|]$ and $j\in[1 \dd |\mt G_i|]$, return the $j$-th edge of $\mt G_i$ in $\mt O(1)$ time, and
    \item For $i\in [1 \dd |\mt H|]$ and $e\in \mt G_i$, return the rank of $e$ in $\mt G_i$ in $\mt O(\log \log |V| + \log \log |\mt H|)$ time.
\end{alphaenumerate}
\end{lemma}

\begin{proof}
The proof is organized in two parts: data-structure and queries implementation.

\subparagraph{Data structure.}

% Consider the following data structures that occupy $\mt O(|E|+|P|\log|\mt H|)$ bits of space:
Let $\mt P$ be a pebbling of $G$ and let $\guardianCollection$ be its induced guardian collection. Our data structure for representing $\mt P$  consists of:
\begin{enumerate}
    \item The array $P[1 \dd |\mt P|] = \ms G(1) \cdot \ms G(2) \cdots \ms G(|E|)$,
    \item The bit vector $B[1 \dd |E| + |\mt P| + 1] = 0 \cdot 1^{|\ms G(1)|} \cdot 0 \cdot 1^{|\ms G(2)|} \cdots 0 \cdot 1^{|\ms G(|E|)|} \cdot 0$, and
    %\item Structures to support $\select_c$ and $\rank_c$ over $B$ for $c\in \{0, 1\}$, and $\select_i$ and $\rank_i$ over $P$
        %for $i = 1, 2, \dots, |\mt H|$.
    \item A succinct rank/select structure over $B$ supporting
$\rank_c$ and $\select_c$ for each $c\in\{0,1\}$, and a rank/select
structure over $P$ supporting $\rank_i$ and $\select_i$ for each
$i\in[1\dd |\mt H|]$. 
%\gyn{The last item has been polished.}
\end{enumerate}

Specifically, the array $P$ consists of the concatenation of $\guardianCollection(1), \dots, \guardianCollection(|E|)$ where every block $\guardianCollection(i)$ is
sorted increasingly. The bitvector $B$ encodes pebble multiplicities: the number of ones between the $j$-th and 
$(j+1)$-th zero equals the number of pebbles on edge $j$. Overall, this data structure consumes
$\mt O(|E| + |\mt P| \log |\mt H|)$ bits of space.

\subparagraph{Queries.}

\emph{Queries (a) and (b).}
Let $e \in E$ be the query edge and $p = \mt R(e)$, $j = \rank_1(B, \select_0(B, p)) + 1$
and $j' = \rank_1(B, \select_0(B, p+1))$. All these operations take
$O(\log\log|V|)$ time. Note that $e$ is pebbled if and only if $j \leq j'$ (query (a))
and, by definition of $P$, the list $\mt G(e)$ is stored consecutively in $P[j \dd j']$ (query (b)). Hence, 
query (a) can be supported in $\mt O(\log\log |V|)$ time and query (b) can be supported with the additional
$\mt O(|\mt G(e)|)$ time to scan $P[j \dd j']$.

\emph{Query (c).} 
The number of pebbles of color $i$ is given by $\rank_i(P, |\mt P|)$, computable in $\mt O(\log \log |\mt H|)$ time.  

\emph{Query (d).}
Let $i \in [1 \dd \mt |H|]$ and $j \in [1 \dd |\mt G_i|]$. Remember that edges are sorted according to 
the topological order of $V$, and $P$ concatenates the colored pebbles according to this order. 
Thus, if $x, y \in [1 \dd |\mt P|]$ are such that $P[x] = P[y] = i$, and $e_x$ and $e_y$ are the 
corresponding edges, then $x < y$ if and only if $e_x < e_y$. Therefore, the $j$-th edge pebbled with
color $i$, say $e_j$, corresponds to the $e_j$-th block of 1s of bitvector $B$.

Let $p_j$ be the position in $P$ of such $j$-th edge with color $i$. Clearly $p_j = \select_i(P, j)$,
which takes $\mt O(1)$ time. Furthermore, the $p_j$-th bit set to 1 in $B$ is in the $e_j$-th group
of 1s in $B$. Therefore, $e_j = \rank_0(B, \select_1(B, p_j))$, which takes $\mt O(1)$ time.
Finally, returning the actual edge $\edgeRank\inv(e_j)$ by Proposition~\ref{prop-graph-rep} also takes $\mt O(1)$ time.

\emph{Query (e).}
Let $i \in [1 \dd |\mt H|]$ and $e \in \mt G_i$. Let $p = \mt R(e)$, which
takes $O(\log\log|V|)$ time, and $j = \rank_1(B, \select_0(B, p))$, which takes $\mt O(1)$ time, and notice that $j+1$ corresponds to the position in
$P$ of the first element of $\ms G(e)$. Therefore, the rank of $e$ in
$\mt G_i$ is equivalent to the number of occurrences of $i$ in $P[1 \dd j]$, that is $\rank_i(P, j)+1$ which takes $\mt O(\log\log|\mt H|)$ time.

\emph{Managing out-of-bound inputs.}
Queries (d) and (e) assume that the inputs $j$ and $e$ satisfy
$j \in [1 \dd |\mt G_i|]$ and $e \in \mt G_i$. Here we prove that this is
not a loss of generality, since the out-of-bound can be detected and
addressed without worsening the time complexity.

During query (d), $j > |\mt G_i|$ if and only if $p_j = \select_i(P, j) = |\mt P|+1$.

During query (e), let $j = \rank_1(B, \select_0(B, \mt R(e)))$ and $j' = \rank_1(B, \select_0(B, \mt R(e)+1))$. Then $e \notin \mt G_i$ if and only if
$\rank_i(P, j) = \rank_i(P, j')$.
\end{proof}

%\textcolor{blue}{Bri: If we have some space left, I think it would be good to add a sketch of how to support the
%less-obvious queries (c), (d), (e).}

\section{Solving \MINSP{}}
\label{sect-pebbling-to-compute}

The main result of this section is  Theorem \ref{theorem-local-for-all} stating the complexity of two different algorithms for solving \MINSP{}. Both algorithms rely
on the idea of determining if there exist multiple paths from one vertex to another.
The first one precomputes a table based solely on the topology of the graph, and
then uses this table for each path in $\mt H$. The second algorithm, instead, while traversing
each path $h \in \mt H$ it computes the information needed to determine whether there exist multiple paths in $G$
between vertices of $h$.

\begin{theorem}\label{theorem-local-for-all}
    Let $G=(V, E, \mt H)$ be a variation graph. A minimum-size pebbling can be computed in time
    $\mt O\!\left(\min\left\{|V|\cdot|E|+\sum |V_i|, \; |\mt H|\cdot|E|\right\}\right) \subseteq \mt O((|V|+|\mt H|)\cdot|E|)$.
\end{theorem}

In the next sections we give a brief description of the algorithms, whose pseudocode are given in Algorithm~\ref{algo:poly1}
and Algorithm~\ref{algo:poly2}.

\subsection{The first algorithm}

\begin{algorithm}
\caption{Compute matrix $A$ that checks if there exist multiple paths.}\label{algo:dp}
\begin{algorithmic}[1]
    \nolinenumbers
    \Function{ComputeA}{$V, E$}
        \State $A[1 \dd |V|, 1 \dd |V|] \gets$ matrix initialized by $0$
        \State $A[|V|, |V|] \gets 1$
        \For{$u = |V|-1$ down to $1$}
            \State $A[u,u] \gets 1$
            \For{$v = u+1$ to $|V|$}
                \For{$x \in \out(u)$}
                    \State $A[u, v] \gets A[u, v] + A[x, v]$
                \EndFor
                \If{$A[u, v] > 2$}
                    \State $A[u, v] \gets 2$
                \EndIf
            \EndFor
        \EndFor
        \State\Return $A$ 
        %\gyn{To be consistent with your algorithm, I change the definition of $\out(u)$ as a list of vertices.}\bri{Super.}
    \EndFunction
\end{algorithmic}
\end{algorithm}

\begin{algorithm}
\caption{First polynomial-time algorithm for solving \MINSP.}\label{algo:poly1}
\begin{algorithmic}[1]
\nolinenumbers
    \Function{PebblePath}{$A, h_i$}
        \State $V_i \gets $ array of vertices $[v_1, \dots, v_k]$ of $h_i$
        \State $\mt P_i \gets \emptyset$
        \State $pre \gets t_i$
        \For{$j = k-1$ down to $1$}
            \State $v \gets V_i[j]$
            \If{$A[v, pre] = 2$}\Comment{Multiple paths from $v_j$ to $pre$.}
                \State Prepend $((v, V_i[j+1]), i)$ to $\mt P_i$
                \State $pre \gets v$
            \EndIf
        \EndFor
        \State\Return $\mt P_i$
    \EndFunction
\end{algorithmic}
\end{algorithm}

% \gyn{Maybe, we may remove the title "High-level idea" and regard the paragraph followed as the description of the algorithm. Otherwise, reviewers would question why not give the detailed description of the algorithm.}
\paragraph*{Algorithm.}

We precompute a matrix $A \in \{0,1,2\}^{V \times V}$, where $A_{u,v}$ is the number of paths in $G$ from $u$ to $v$ (where this number
is limited to 2 if it is larger). 
Algorithm \ref{algo:dp} works by filling the matrix $A$ by visiting the graph following the topological  order from the vertex preceding $|V|$ down to the vertex $1$. In this way, the needed entries of the recurrence to fill $A$  are already pre-computed. 
Intuitively, a pebble is required whenever a vertex $u$ has multiple outgoing paths to the next ``checkpoint'' $v$, as this multiplicity violates the unique path condition required for reconstruction.
To compute $A$, we adapt the standard dynamic programming (DP) algorithm for counting paths between all pairs of vertices in a DAG, whose pseudocode is given in Algorithm~\ref{algo:dp}.
%\gyn{I think we could add more details in how to compute $A$. For the details, you could reference to the appendix of the stacs version.}
% \bri{Isn't this folklore? It's the most basic DP on DAGs.}

Afterwards, we compute a sub-pebbling $\mt P_i$ inducing a minimum-size guardian $\mt G_i$ for $h_i \in \mt H$.  
Two variables are used: $\mt P_i$, the sub-pebbling that we are constructing, and $pre$, the starting vertex of the most recently pebbled edge.  
Initially, set $pre = t_i$ and $\mt P_i = \emptyset$.  
We traverse $V_i = (v_1, \dots, v_k)$ backwards, from $v_{k-1}$ down to $v_1$ (remember that $v_1 = s_i$ and $v_k = t_i$).
For every $j$, check whether there is a unique path in $G$ from $v_j$ to $pre$ (\emph{i.e.}, $A_{v_j, pre} = 1$).
If not, edge $(v_j, v_{j+1})$ is pre-pended to $\mt P_i$ (that is, it is pebbled), and we set $pre = v_j$.  
After the traversal, $\mt P_i$ is returned. 
In the end, the set $\mt P = \bigcup \mt P_i$ is a minimum-size pebbling for $G$.

\paragraph*{Correctness.}

Since there is no interplay between guardians of different colors, a minimum size pebbling
can be obtained by independently minimizing the size of the guardians. Thus, correctness follows
from the following lemma.

\begin{lemma}\label{lem:corr-poly1}
    Let $\mt G_i$ be the edges pebbled by color $i$ from Algorithm~\ref{algo:poly1}. 
    Then, $\mt G_i$ is a minimum-size guardian for $h_i$. 
\end{lemma}
\begin{proof}
    Let $\mt G_i = (e_1, e_2, \dots, e_n)$. By construction, $(s_i, \mt G_i, t_i)$ satisfies
    the unique path condition, hence $\mt G_i$ is a guardian for $h_i$. 

    Suppose, for the sake of contradiction, that there exists a guardian $\mt G' = (e'_1, e'_2, \dots, e'_m)$
    for $h_i$ of size $m < n$. Append $(t_i, t_i)$ to both $\mt G_i$ and $\mt G'$ as their last element, and
    assume the rank of $(t_i,t_i)$ to be the largest among all edges in $h_i$.
    We also let $e_j = (u_j, v_j)$ and $e'_j = (u'_j, v'_j)$ for every valid $j$. 

    We first prove by induction that for all $0 \leq j \leq m$ it holds $e_{n+1-j} \leq e'_{m+1-j}$.

    (Base) For $j = 0$ we have $e_{n+1} = (t_i, t_i) = e_{m+1}$.
    
    (Step) Assume the claim holds for $j$ and consider the case for $j+1$. By Algorithm~\ref{algo:poly1}
        there exist multiple paths going from $u_{n-j}$ to $u_{n+1-j}$. Moreover, since by induction
        hypothesis $e_{n+1-j} \leq e'_{m+1-j}$, there are multiple paths going from $u_{n-j}$ to $u'_{m+1-j}$.
        If it were $e'_{m-j} < e_{n-j}$ we would have multiple paths going from $u'_{m-j}$ to $u'_{m+1-j}$,
        contradicting the fact that $\mt G'$ is a guardian. Therefore, $e_{n-j} \leq e'_{m-j}$.

    Now consider $e_{n+1-m}$ and $e'_{m+1-m} = e'_1$. By the above claim we have $e_{n+1-m} \leq e'_1$.
    Since we assumed $m < n$, it holds $e_1 < e_{n+1-m}$ and by Algorithm~\ref{algo:poly1} we have
    multiple paths going from $s_i$ to $u_{n+1-m}$ and, so, also to $u'_1$, contradicting the unique
    path condition on $\mt G'$.
\end{proof}

\paragraph*{Complexity.}
The matrix $A$ is computed in $\mt O(|V|\cdot|E|)$ time, and it can be used for every $h_i \in \mt H$. 
Once we have computed $A$, processing each vertex of $V_i$ takes $\mt O(1)$ time. Thus, in total, a minimum-size pebbling 
can be constructed in $\mt O\!\left(|V|\cdot|E|+\sum |V_i|\right) \subseteq \mt O((|V|+|\mt H|)\cdot|E|)$ time.

\subsection{The second algorithm}

\begin{algorithm}
\caption{Second polynomial-time algorithm for solving \MINSP.}\label{algo:poly2}
\begin{algorithmic}[1]
\nolinenumbers
    \Function{PebblePath}{$h_i$}
        \State $V_i \gets $ array of vertices $[v_1, \dots, v_k]$ of $h_i$
        \State $\mt P_i \gets \emptyset$
        \State $pre \gets t_i$
        \State $pre.paths \gets 1$
        \State $pnt \gets k$\Comment{Points to the most recently discovered element of $V_i$.}
        \For{$u = t_i-1$ down to $s_i$}
            \State $u.paths \gets 0$
            \For{$v \in \out(u)$}\Comment{Number of paths from $u$ to $pre$.}
                \If{$v \leq pre$}
                    \State $u.paths \gets \min\{\,2, u.paths + v.paths\,\}$
                \EndIf
            \EndFor
            \If{$u = V_i[pnt-1]$}\Comment{New edge to be pebbled.}
            \If{$u.paths = 2$}
                \State $v \gets V_i[pnt]$\Comment{Unique vertex in $\out(u) \cap V_i$.}
                \State Prepend $((u,v), i)$ to $\mt P_i$
                \State $pre \gets u$
                \State $pre.paths \gets 1$
            \EndIf
                \State $pnt \gets pnt-1$
            \EndIf
        \EndFor
        \State\Return $\mt P_i$
    \EndFunction
\end{algorithmic}
\end{algorithm}

\paragraph*{Algorithm.}

Similar to the first algorithm, the second one also computes a sub-pebbling $\mt P_i$ inducing a minimum-size guardian $\mt G_i$ for 
each $h_i \in \pathSet$ independently, with the main difference being that, instead of precomputing the information about multiple 
paths between pairs of vertices $u,v \in V$, we compute the same information (on-the-fly) for vertices $u \in V$ relative to the 
starting vertex $pre \in V_i$ of the most recently pebbled edge in $h_i$ (we still process the graph from right to left).

Given $h_i \in \mt H$, we create the sub-pebbling $\mt P_i$ that will store the pebbled edges of $h_i$, and a variable $pre$ such that
$(pre, x)$ is the most recently pebbled edge of $h_i$, for some $x \in V_i$. Initially, $\mt P_i = \emptyset$ and $pre = t_i$.

The algorithm processes vertices in reverse order, from $t_i-1$ down to $s_i$ (which are the only vertices that may be found on path $h_i$) including those not in $V_i$.  
At each visited vertex $u$, we keep a variable $u.paths$ to store the number of paths in $G$ from $u$ to $pre$ (we limit $u.paths \leq 2$).
It is not difficult to see that $u.paths = \min\{\,2, n_u\,\}$, where $n_u = \sum_{v \in \out(u), v \leq pre} v.paths$.
Indeed, if $v \in \out(u)$ is such that $v \leq pre$ then $v.paths$ is the number of paths in $G$ from $v$ to $pre$ (upper-limited to 2).
In case $u \in V_i$ and $u.paths = 2$, then we have to pebble the unique edge $(u,v)$ outgoing from $u$ that lies on $h_i$, 
and we set $pre = u$ and $pre.paths = 1$ (representing the empty path from $pre$ to itself).

After the traversal, $\mt P_i$ is returned.
In the end, the set $\mt P = \bigcup \mt P_i$ is a minimum-size pebbling for $G$.

%\gyn{If you gave the correctness of the first algorithm, then you don't need to prove the correctness of the algorithm here. At least, I think the correctness of one of the two algorithms should be given.}
% \bri{Done! What should we do with correctness for this algorithm? Shall we briefly refer it to the first algorithm or do we skip it entirely?}
% \gyn{I would suggest to prove that both algorithms output the same pebbling; thereby, the correctness follows from algorithm 1.}
% \gyn{I would suggest proving that the two algorithms output the same pebbling; the correctness of Algorithm~2 would then follow directly from that of Algorithm~1.}\bri{Done.}

\paragraph*{Correctness}

Note that at the iteration for $u$ we correctly compute in $u.paths$ the (bounded) number
of paths going from $u$ to $pre$ (we use the same recurrence of Algorithm~\ref{algo:dp}).
Moreover, whenever $u$ is a vertex on path $h_i$ and there are multiple paths going from
$u$ to $pre$ (that is, $u.paths = 2$) we pebble the unique edge $(u,v)$ that lies on $h_i$.
This is precisely the pebbling condition we use also in Algorithm~\ref{algo:poly1}, hence,
in the end, the pebbling generated by Algorithm~\ref{algo:poly2} will be the same.

\paragraph*{Complexity}

Fix a path $h_i \in \mt H$. The algorithm processes every vertex and every edge of $G$ once. In case $u.paths = 2$, checking whether $u \in V_i$
and finding the unique (if it exists) edge $(u,v)$ to be pebbled can be done in constant time by maintaining a pointer to the most recently discovered
vertex of $V_i$ while processing the vertices backwards. Thus, in total, the algorithm requires $\mt O(|E|)$ time in the worst case.
Summing over all paths $h \in \mt H$, we obtain $\mt O(|\mt H| \cdot |E|)$.

% \gyn{I think we could briefly mention that both algorithms output the same min-size pebbling, although there might be multiple min-size pebbling for a given graph.}
% \bri{Agree.}

As a final remark, notice that there may exist minimum size pebblings different from the one
generated by both Algorithm~\ref{algo:poly1} and Algorithm~\ref{algo:poly2}.

\section{Complexity of \MINWP{} }
\label{sect-NP-hard}

% In this section, we prove that \MINWP{} is NP-hard even when every predefined path has at most nine edges.
% For space constraints, our ILP solution to the \MINWP{} problem is deferred to Appendix \ref{app-mini-weight-pebbling}. 
% %A DAG $G=(V,E)$ is called \emph{transitive-free} if, for every edge $(u,v)\in E$, the unique path from $u$ to $v$ is the edge $(u,v)$ itself. 
% All proofs omitted in this section can be found in Appendix \ref{app-NP-Hard}.

We first show in Section~\ref{sect-np-hard} that computing a minimum-size \emph{saturated} pebbling is \NP-hard, even when every path $h_i \in \mt H$ contains only a constant number of edges. We then show in Section~\ref{sect-set-cover} that \MINWP{} admits a reduction to the minimum-weight set cover problem.
Finally, in Section~\ref{sect-ILP}, we formulate \MINWP{} as an integer linear program.
%All omitted proofs in the section are provided in Appendix \ref{app-NP-Hard}.

%Due to space constraints, the reduction from \MINWP{} to minimum-weight set cover is deferred to Appendix 
%\ref{app-mini-weight-pebbling}. All other omitted proofs are provided in Appendix \ref{app-NP-Hard}.

\subsection{The \NP-Hardness Proof}
\label{sect-np-hard}

% \subsection{A simplified \NP-hard proof \textcolor{red}{Y.G.: TODO}}

\paragraph*{Overview of the reduction.}
We establish the \NP-hardness of \MINWP{} through a reduction from {Vertex Cover} on cubic triangle-free  graphs. First,
Lemma~\ref{prop-transitive-free} shows that the edges of any cubic triangle-free
graph can be oriented in linear time so as to obtain a transitive-free DAG $G$. Since the underlying undirected graph of $G$ is the original cubic triangle-free graph, the two graphs have the same minimum vertex-cover size.
% \pb{\(G\). 
% Then the minimum vertex-cover of the underlying undirected graph of \(G\) is the same of the original cubic triangle-free graph}\gyn{Fixed.} 
Then, from the resulting DAG \(G\), we construct a variation graph \(G'\) by introducing, for each edge of \(G\), the reduction gadget shown in Fig.~\ref{fig-reduction-weight}. 
%Lemma~\ref{lem-me-unique} establishes a property  of the predefined paths in $G'$ that follows from the transitive-free condition imposed on \(G\).
In
Lemma~\ref{lem-eq-saturated-min-degree-weight}, we establish the connection
between the size of a minimum saturated pebbling of \(G'\) and the minimum
total degree of the vertices in a vertex cover of the underlying undirected graph of $G'$. 
% \pb{if you mention here a property you have to state what property otherwise skip this statement}\gyn{The property is deleted from the overview.}
Finally, Theorem~\ref{theorem-non-st-mwp-NP-Hard} gives the complete reduction from
\textsc{Vertex Cover} on cubic triangle-free graphs to \MINWP{} and concludes
that \MINWP{} is \NP-hard, even when every predefined path contains at most
three edges.

\begin{lemma}\label{prop-transitive-free}
There is an $\mt O(|E|+|V|)$-time algorithm that transforms any undirected cubic triangle-free graph $X=(V,E)$ into a transitive-free DAG by orienting the edges of $E$.
\end{lemma}

\begin{proof}
	Since $X$ is a cubic triangle-free, $X$ cannot be either a cycle of odd length or a complete graph.
	By Brooks's theorem~\cite{brooks1941colouring}, the graph $X$ is $3$-colorable; that is, each vertex of $X$ can be assigned one of three colors so that no two adjacent vertices receive the same color.
	
	We first compute a proper vertex coloring 	$c:V\to\{1,2,3\}$.
	By \cite{baetz2014brooks}, a Brooks coloring can be computed in $O(|V|+|E|)$ time.
	Then, for every edge $\{u,v\}\in E$, orient the edge from the endpoint of smaller color to the endpoint of larger color. 
    
    Let $G$ be the resulting directed graph.
	Since the vertex colors strictly increase along every directed edge, they also
	strictly increase along every directed path. Hence, $G$ is acyclic.
	
	It remains to prove that $G$ is transitive-free. Suppose, for the sake of contradiction, that an oriented edge $(u,v)$ is transitive. Then,
	besides the edge $(u,v)$, there exists a directed path from $u$ to $v$ that contains at least two edges.
	However, since the vertex colors along this path strictly increase and only three colors are available, the path consists of exactly two edges. Then, this alternative path and the edge $(u, v)$ together form a triangle in $X$, contradicting the assumption that $X$ is triangle-free.
	
	Thus, no edge of $G$ is transitive, and $G$ is a transitive-free DAG.
	The coloring and orientation can be computed in $O(|V|+|E|)$ time.
\end{proof}

\paragraph*{Construction.}  
%Let $G=(V,E)$ be a transitive-free DAG. 
We first describe the construction for an arbitrary transitive-free DAG $G=(V,E)$.
We denote by $\Delta$ the maximum degree among all vertices in $V$. 
In the reduction of Theorem~\ref{theorem-non-st-mwp-NP-Hard}, $G$ is obtained by orienting a cubic triangle-free graph according to Lemma~\ref{prop-transitive-free}, and therefore $\Delta=3$.

% \pb{Am I wrong or you said that you start the reduction from a cubic transitive free DAG?}\gyn{You are right. The DAG used in the final reduction is obtained from a cubic triangle-free graph, so its underlying undirected graph is cubic and $\Delta=3$. We stated the construction more generally for an arbitrary transitive-free DAG because Lemma~12 do not require the cubic constraint. I have now clarified this point and explicitly noted that $\Delta=3$ in the final reduction.} 

From $G$, we construct an instance $G'=(V',E',\mt H)$ of the pebbling problem by adding elements to $V'$ and $E'$ based on the gadget detailed 
in Fig.~\ref{fig-reduction-weight}. 
Specifically, for every vertex $u \in V$ we create two vertices $\underline{u}$ and $\overline{u}$, and an edge
$\mf n_u = (\underline{u}, \overline{u})$ called the \emph{node edge}; for every directed edge $e =(u,v) \in E$ we create $2\Delta$ vertices 
$u^e_1, u^e_2, \dots,  u^e_\Delta$, $v^e_1, v^e_2, \dots, v^e_\Delta$, and the following $2(\Delta+1)$ edges:  
\begin{itemize}
    \item The \emph{arc edge} $\mf a_{u,v} = (\overline{u},\underline{v})$,
    \item The \emph{shortcut edge} $\mf s_{u,v} = (\underline{u},\overline{v})$, and
    \item The \emph{fake edges} $\overline{\mf f^e_i} = (u^e_i,\overline{u})$ and $\underline{\mf f^e_i} = (\underline{v}, v^e_i)$,
        for each $1 \leq i \leq \Delta$.
\end{itemize}
In the following, we will refer to a generic node (resp. arc, shortcut, fake) edge as an $\mf n$- (resp. $\mf a$-, $\mf s$-, $\mf f$-) edge.
The set $\mt H$ contains the following paths:
\begin{itemize}
    \item The \emph{shortcut paths} $\mf S_{u,v} = (\mf s_{u,v})$, for every $(u, v) \in E$,
    \item The \emph{arc paths} $\mf A_{u,v} = (\mf n_u, \mf a_{u,v}, \mf n_v)$, for every $(u, v) \in E$, and
    \item The \emph{fake paths} $\mf F^e_i = (\overline{\mf f^e_i}, \mf a_{u,v}, \underline{\mf f^e_i})$ for all $e \in E$ and $1 \leq i \leq \Delta$.
\end{itemize}
The symbols $\mf S_{u,v}$, $\mf A_{u,v}$, and $\mf F^e_i$ will be used as pebble colors later.

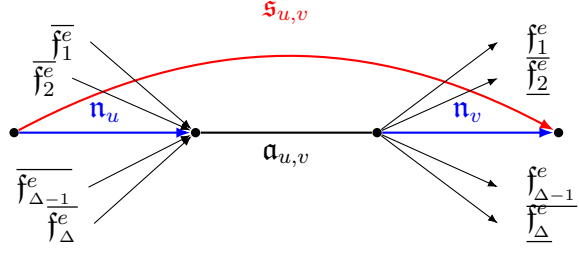
\begin{figure}%[!h]
    \centering
    \usetikzlibrary{calc,arrows.meta} % For coordinates and custom arrows

\begin{tikzpicture}[scale=.8]
  \tikzset{vertex/.style={circle, fill=black, minimum size=3.5pt, inner sep=0pt}}
  \tikzset{blue edge/.style={draw=blue, thick, -{Latex[scale=0.8]}}}
  \tikzset{red edge/.style={draw=red, thick, -{Latex[scale=0.8]}}}
  \tikzset{aux edge/.style={draw=black, -{Latex[scale=0.8]}}}

  \coordinate (n1_pos) at (0,0);
  \coordinate (n2_pos) at (3.0,0);
  \coordinate (n3_pos) at (6.0,0);
  \coordinate (n4_pos) at (9.0,0);

  \node[vertex] (n1) at (n1_pos) {};
  \node[vertex] (n2) at (n2_pos) {};
  \node[vertex] (n3) at (n3_pos) {};
  \node[vertex] (n4) at (n4_pos) {};

  \draw[blue edge] (n1) -> (n2) node[midway, above, text=blue, font=\large] {$\mf n_u$};
  \draw[thick] (n2) -> (n3) node[midway, below, font=\large] {$\mf a_{u,v}$};
  \draw[blue edge] (n3) -> (n4) node[midway, above, text=blue, font=\large] {$\mf n_v$};
  \draw[red edge, bend left=28] (n1) to node[midway, above, text=red, yshift=8pt, font=\large] {$\mf s_{u,v}$} (n4);

  \def\laboffset{0.1}
  \node[font=\large] (u1_l) at ($(n2)+(-2.2, 1.5)$) {$\overline{\mf f^e_1}$};
  \node[font=\large] (u2_l) at ($(n2)+(-2.5, 0.9)$) {$\overline{\mf f^e_2}$};
  \node[font=\large] (ud_m1_l) at ($(n2)+(-2.5, -0.9)$) {$\overline{\mf f^e{_{_{\!\!\!{\Delta-1}}}}}$};
  \node[font=\large] (ud_l) at ($(n2)+(-2.2, -1.5)$) {$\overline{\mf f^e{_{_{\!\!\!\Delta}}}}$};
  
  \foreach \i in {u1_l, u2_l, ud_m1_l, ud_l} {
    \draw[aux edge] ($(\i.east)+(\laboffset,0)$) -- (n2);
  }

  \def\edgedest{2.0}
  \coordinate (v_dest_1) at ($(n3)+(\edgedest, 1.5)$);
  \coordinate (v_dest_2) at ($(n3)+(\edgedest, 0.9)$);
  \coordinate (v_dest_3) at ($(n3)+(\edgedest, -0.9)$);
  \coordinate (v_dest_4) at ($(n3)+(\edgedest, -1.5)$);
  
  \foreach \i in {1,2,3,4} {
    \draw[aux edge] (n3) -- (v_dest_\i);
  }
  
  \def\labelright{0.3}
  \node[font=\large, right] at ($(v_dest_1)+(\labelright,0)$) {$\underline{\mf f^e_1}$};
  \node[font=\large, right] at ($(v_dest_2)+(\labelright,0)$) {$\underline{\mf f^e_2}$};
  \node[font=\large, right] at ($(v_dest_3)+(\labelright,0)$) {$\underline{\mf f^e_{_{\!{\Delta-1}}}}$};
  \node[font=\large, right] at ($(v_dest_4)+(\labelright,0)$) {$\underline{\mf f^e_{_{{\!\Delta}}}}$};

\end{tikzpicture}
    \caption{The reduction gadget associated with the edge $e = (u, v)\in E$. It consists of the node edges $\mf n_u$ and $\mf n_v$, the arc edge $\mf a_{u, v}$, the shortcut edge $\mf s_{u, v}$, and $\Delta$ incoming and $\Delta$ outgoing fake edges. The corresponding predefined paths are the shortcut path $(\underline{u}, \mf s_{u, v}, \overline{v})$, the arc path $(\underline{u},\mf n_u, \mf a_{u,v}, \mf n_v,\overline{v})$, and the fake paths $(u^e_i, \overline{\mf f^e_i}, \mf a_{u,v}, \underline{\mf f^e_i},v^e_i)$ for each $i\in[1..\Delta]$.} 
    \label{fig-reduction-weight}
\end{figure}

Since $G$ is a DAG, the constructed graph $G'$ is also a DAG. Most importantly, $G'$ is of polynomial size \emph{w.r.t.} $G$, as it 
contains $2|V| + 2\Delta|E|$ vertices, $2(\Delta+1)|E| + |V|$ edges, and the set $\mt H$ consists of $(2+\Delta)|E|$ paths.
The instance $G' = (V', E', \mt H)$ can be constructed in polynomial time.
By construction, every edge in $E'$ is traversed by at least one path of $\mt H$.  
Moreover, it follows that $|\col(\mf n_u)|=\deg(u)\le \Delta$, $|\col(\mf a_{u, v})|=\Delta+1$, and $|\col(\mf s_{u, v})|=1$.

We will always refer to $G = (V, E)$ as a transitive-free DAG and to $G' = (V', E', \mt H)$ as the variation graph constructed from $G$.

\begin{lemma}\label{lem-me-unique}
    Let $(u, v) \in E$. Then, (1) $(\mf n_u, \mf a_{u,v})$ is the unique path in $G'$ from $\underline{u}$ to $\underline{v}$, and 
    (2) $(\mf a_{u,v}, \mf n_v)$ is the unique path in $G'$ from $\overline{u}$ to $\overline{v}$.
\end{lemma}
\begin{proof}
We prove the first statement; the second is symmetric.

Let $\pi = (e_1, \dots, e_k)$ be a different path going from $\underline{u}$ to $\underline{v}$.
Clearly, no edge of $\pi$ can be an $\mf f$-edge, since fake edges are incident with private sources (that have no incoming edge) or private sinks (that have no outgoing edge).

Replace every shortcut edge $\mf s_{x, y}$ with the path $(\mf n_x, \mf a_{x,y}, \mf n_y)$.
The resulting path $\pi'$ contains only node and arc edges. By construction, node and arc edges alternate in this path.
Note that every node edge $\mf n_x$ corresponds to a node $x\in V$ and that every arc edge $\mf a_{x,y}$ corresponds to an edge $(x, y)\in E$.
As a result, the path $\pi'$ in $G'$ projects to a directed path in $G$ from $u$ to $v$.

Since $(u, v)\in E$ and $G$ is transitive-free, the only directed path in $G$ from $u$ to $v$ is the edge $(u, v)$.
Therefore, the resulting path $\pi'$ must be $(\mf n_u, \mf a_{u,v})$.
%Since $\pi'$ contains less than three edges, $\pi'$ must be $\pi$, contradicting the choice of $\pi$. 
In particular, \(\pi'\) contains exactly two edges. Replacing a shortcut
edge by \((\mf n_x,\mf a_{x,y},\mf n_y)\) increases the number of edges
by two, so \(\pi\) cannot contain any shortcut edge. Hence,
\(\pi=\pi'=(\mf n_u,\mf a_{u,v})\), contradicting the choice of \(\pi\).

The second statement follows analogously.
\end{proof}

% \pb{you never use U(G) in other parts of the file...why to use it here? You named it X if I am not wrong, otherwise please use it below in the proof...}
% \gyn{Since the notation $U(G)$ is not otherwise used, I have removed it and now refer directly to the underlying undirected graph of $G$ throughout. Note that $X$ used in Lemma~10 denotes any cubic triangle-free graph, whereas Lemma~12 is stated for any arbitrary transitive-free DAG.}

% Let \(U(G)\) denote the underlying undirected graph of \(G\).
For each vertex $u\in V$, let $\deg_G(u)$ denote the degree of $u$ in the underlying undirected graph of $G$.
For every vertex cover $C$ of the underlying undirected graph of $G$, define its \emph{degree weight} by $w_{\deg}(C)
=
\sum_{u\in C} \deg_G(u).$
We then define the minimum degree weight of a vertex cover of the underlying undirected graph of $G$ as
$w_{\deg}^*(G)
=
\min\bigl\{
w_{\deg}(C)
:
C \text{ is a vertex cover of the underlying undirected graph of } G
\bigr\}.$

\begin{lemma}\label{lem-eq-saturated-min-degree-weight}
Let \(G=(V,E)\) be a transitive-free DAG, and let \(G'=(V',E', \mt H)\) be the variation
graph obtained from \(G\) by the construction above, and let
\(\mt P^*\) be a minimum saturated pebbling of \(G'\). Then, $|\mt P^*|=|E|+w^*_{\deg}(G).$
\end{lemma}

\begin{proof}
It is sufficient to prove the two inequalities $|\mt P^*|\le |E|+w^*_{\deg}(G)$ and $|\mt P^*|\ge |E|+w^*_{\deg}(G)$.

To prove that $|\mt P^*|\le |E|+w^*_{\deg}(G)$, let $\mt V \subseteq V$ be a vertex cover of the underlying undirected graph of $G$ such that $w_{\deg}(\mt V)=w^*_{\deg}(G)$.
We construct a saturated pebbling $\mt P$ for $G'$ based on $\mt V$. To simplify the exposition,
when we say that an edge $e$ is added to $\mt P$ we intend that every pebble $(e, i)$ for $i \in \col(e)$
is added to $\mt P$ (hence, saturation is satisfied by construction).
\begin{itemize}
    \item For every shortcut path $\mf S_{u,v} \in \mt H$, add pebble $(\mf s_{u,v}, \mf S_{u,v})$ to $\mt P$.
        Clearly, $(\underline{u}, \{\mf s_{u,v}\}, \overline{v})$ satisfies the unique path property for $\mf S_{u,v}$.
        This case will add $|E|$ elements to $\mt P$.

    \item For every fake path $\mf F^e_i \in \mt H$, add no pebbles to $\mt P$.
        Since there is a unique path in $G'$ from $\be(\mf F^e_i)$ to $\en(\mf F^e_i)$, the empty
        set is a valid guardian for $\mf F^e_i$.

    \item For every arc path $\mf A_{u,v} \in \mt H$, add to $\mt P$ the edge $\mf n_u$ if $u \in \mt V$
        and the edge $\mf n_v$ if $v \in \mt V$ (both edges are added if both vertices are in the vertex cover).
        Since $\mt V$ is a vertex cover, for every edge $(u,v) \in E$  at least  $u \in \mt V$ or $v \in \mt V$, or both $u, v \in V$.
        If $u \in \mt V$, then by (2) of Lemma~\ref{lem-me-unique} there exists a unique path in $G'$ from
        $\overline{u}$ to $\overline{v}$, hence $(\underline{u}, \{\mf n_u\}, \overline{v})$ satisfies
        the unique path condition. A similar argument holds if $v \in \mt V$ (or both).
        This case will add $\sum_{u \in \mt V} |\col(\mf n_u)| = \sum_{u \in \mt V} |\deg(u)|=w_{\deg}(\mt V)=w^*_{\deg}(G)$ elements to $\mt P$.
\end{itemize}
Therefore, the size of the pebbling is $|\mt P| = w^*_{\deg}(G) + |E|$.
Since $|\mt P^*|\le|\mt P|$, we have $|\mt P^*|\le w^*_{\deg}(G) + |E|$.

%To prove that $|\mt P^*|\ge |E|+w^*_{\deg}(G)$, let $\mt E \subseteq E'$ be the set of edges that appear in $\mt P^*$.

We now prove the lower bound. Let \(\mt E^*\subseteq E'\) be the set of
edges carrying at least one pebble in \(\mt P^*\). Since \(\mt P^*\)
is saturated,
$|\mt P^*|
=
\sum_{e\in\mt E^*}|\col(e)|$.

\begin{itemize}
    \item $\mt E^*$ contains no fake edges. 
     A fake edge belongs only to its corresponding fake path, and that path admits the empty guardian, as shown above. Hence, removing a selected fake edge preserves feasibility and strictly decreases the size of the pebbling, contradicting the minimality of \(\mt P^*\).
    %Since each fake path already has a unique path between its endpoints, removing fake edges preserves the unique path condition,
    \item $\mt E^*$ contains every shortcut edge. Indeed, the empty set cannot be a guardian for the shortcut path \(\mf S_{u,v}\), because
there are at least two paths from \(\underline u\) to \(\overline v\). Since \(\mf S_{u,v}\) consists only of \(\mf s_{u,v}\), this edge must
be selected.  
    \item $\mt E^*$ contains no arc edges by Claim~\ref{claim-no-arc-edge}.
    % Indeed, if there are any $\mf a_{u,v}\in \mt E$, replacing $\mf a_{u,v}$ with $\mf n_u$ (or symmetrically by $\mf n_v$) would preserve the unique path condition by Lemma~\ref{lem-me-unique} meanwhile strictly decrease the pebbling size, since $|\col(\mf n_u)|=\deg(u)<\Delta+1=|\col(\mf a_{u, v})|$.  This contradicts the minimality of \(\mt P^*\). Hence, no arc edge is selected.
\end{itemize}

\begin{claim}\label{claim-no-arc-edge}
 $\mathcal E^*$ contains no arc edge.
\end{claim}

\begin{claimproof}
Suppose, for the sake of contradiction, that
$\mf a_{u,v}\in \mt  E^*$ for some $(u,v)\in E$.
Let $\mt E'=(\mt E^*\setminus\{\mf a_{u,v}\})\cup\{\mf n_u\}$,
and let $\mt P'$ be the saturated pebbling obtained by placing,
for every edge $e\in\mt E'$, all pebbles $(e,i)$ with
$i\in\col(e)$.
We prove that $\mt P'$ induces guardians for $\mt H$.

The only paths whose induced guardians may be affected are the paths containing $\mf a_{u,v}$ or $\mf n_u$.
Each fake path containing $\mf a_{u,v}$ admits the empty
guardian, since there is a unique path from its starting vertex to its
ending vertex. Hence, removing $\mf a_{u,v}$ preserves the guardian property for all such fake paths.

Consider now the arc path $\mf A_{u,v}=(\mf n_u, \mf a_{u,v}, \mf n_v).$ After the replacement, $\mf n_u$ belongs to its guardian.
If $\mf n_v$ is not selected, then the subpath  $(\mf a_{u,v}, \mf n_v)$ of $\mf A_{u,v}$  is the unique path between its
endpoints by Lemma~\ref{lem-me-unique}.
If $\mf n_v$ is selected, then the subpath strictly between $n_u$ and $n_v$, that is the edge $\mf a_{u,v}$, is also unique.
Thus, the induced guardian for $\mf A_{u,v}$ remains valid.

Finally, for every other arc path containing $\mf n_u$, the replacement only adds $\mf n_u$ to its guardian, which cannot
destroy the guardian property. All remaining paths are unchanged.
Therefore, $\mathcal P'$ is a saturated pebbling.

On the other hand, $|\mt P'| \leq |\mt P^*| - |\col(\mf a_{u,v})|+|\col(\mf n_u)|=|\mt P^*|-(\Delta+1)+\deg(u)<|\mt P^*|$, since $\deg(u)\leq\Delta$.
This contradicts the minimality of $\mathcal P^*$, so $\mathcal E^*$ contains no arc edge.
\end{claimproof}

Define $\mt V = \{u \in V : \mf n_u \in \mt E^* \}$.
For any edge $(u,v) \in E$, since there exist at least two paths in $G'$ going from $\underline{u}$ to $\overline{v}$ and $\mf a_{u, v}\notin \mt E^*$, at least one of
$\mf n_u$ and $\mf n_v$ is in $\mt E^*$, thus, at least one of  $u$ and $v$ is in $\mt V$. Hence, $\mt V$ is a vertex cover of the underlying undirected graph of $G$.

Since $\mt E^*$ contains every shortcut edges and $\sum_{u \in \mt V} |\col(\mf n_u)| = \sum_{u \in \mt V} |\deg(u)|=w_{\deg}(\mt V)$, the size $|\mt P^{*}|=|E|+w_{\deg}(\mt V)$.
Since $w_{\deg}(\mt V)\ge w^*_{\deg}(G)$, we have $|\mt P^{*}|\ge|E|+w^*_{\deg}(G)$.

Combining both inequalities, it follows that $|\mt P^{*}|=|E|+w^*_{\deg}(G)$.
\end{proof}

\begin{theorem}\label{theorem-non-st-mwp-NP-Hard}
 \MINWP{} on DAGs is \NP-hard, even when every predefined path has at most three edges.
\end{theorem}

\begin{proof}
    We give a polynomial-time reduction from {Vertex Cover} on cubic
triangle-free graphs, which is \NP-hard as an immediate consequence
of~\cite[Theorem 3]{miRCCW98}.\footnote{The cited theorem explicitly proves
that \textsc{Independent Set} is \NP-complete on triangle-free \(3\)-regular
graphs. The claimed hardness follows from the standard equivalence between
{Vertex Cover} and {Independent Set}.}
    Let \((X,k)\) be an instance of {Vertex Cover}, where \(X=(V,E)\) is a cubic triangle-free graph. The question is whether \(X\) admits a vertex cover of size at most \(k\).

    By Lemma~\ref{prop-transitive-free}, we can orient the edges of $X$ in $\mt O(|E|+|V|)$ time to obtain a transitive-free DAG $G=(V, E)$.
    Since this operation only assigns a direction to each edge of $X$, the underlying undirected graph of $G$ is precisely $X$.
    Therefore, the underlying undirected graph of $G$ and $X$ have the same vertex covers and, in particular, the same minimum vertex-cover size.
    % \pb{the following statement is wrong because $G$ is a DAG...you should say that the minimum vertex cover of the underlying undirected graph of $G$ is the same of the graph $X$} \gyn{Addressed.}

    From $G$, we construct the variation graph $G'=(V',E',\mt H)$, as described above.
    Since $X$ is cubic, the maximum degree $\Delta$ is three, a constant, and thus the construction can be completed in polynomial time.
    Moreover, every path in $\mt H$ contains at most three edges.

    % Let $\mt P^*$ denote a minimum-size saturated pebbling of $G'$.
    % By Lemma~\ref{lem-eq-saturated-min-degree-weight}, we have $|\mt P^*|=|E|+w^*_{\deg}(G)=|E|+w^*_{\deg}(X)$, since $X$ is the underlying undirected graph of $G$.
    % Let $\tau(X)$ denote the minimum size of a vertex cover of $X$.
    % Since $X$ is cubic, we have $\deg(u)=3$ for every $u\in V$, and thus $|\mt P^*|=|E|+w^*_{\deg}(X)=|E|+3\tau(X)$.

    % Since the source graph $X$ is cubic, a vertex cover of $X$ with at most $k$ vertices has degree weight at most $3k$, and conversely any vertex cover of degree weight at most $3k$ contains at most $k$ vertices.
    % Since $|\mt P^*|=|E|+3\tau(X)$, $X$ has a vertex cover of size at most $k$, that is $\tau(X)\le k$, if and only if $|\mt P^*|=|E|+3\tau(X)\le |E|+3k$.

    % This concludes the polynomial-time reduction from Vertex Cover on cubic triangle-free  graphs to \MINWP{}.

    Let \(\tau(X)\) denote the minimum size of a vertex cover of \(X\). Since every vertex of \(X\) has degree three, every vertex cover \(C\) satisfies $w_{\deg}(C)=3|C|$. Hence, $w_{\deg}^*(G)=3\tau(X)$. By Lemma~\ref{lem-eq-saturated-min-degree-weight}, if \(\mt P^*\) is a minimum saturated pebbling of \(G'\), then
$|\mt P^*|
=
|E|+w_{\deg}^*(G)
=
|E|+3\tau(X).$
Therefore, \(X\) admits a vertex cover of size at most \(k\) if and only if $\tau(X)\leq k$,
which holds if and only if
$|\mt P^*|
=
|E|+3\tau(X)
\leq
|E|+3k.$
Equivalently, \(X\) admits a vertex cover of size at most \(k\) if and only
if \(G'\) admits a saturated pebbling of size at most \(|E|+3k\).

This proves that the decision version of \MINWP{} is \NP-hard. It belongs
to \NP, since a saturated pebbling provides a polynomial-size certificate
whose saturation and unique-path conditions can be verified in polynomial
time. Hence, the decision version is \NP-complete.
Consequently, the optimization problem \MINWP{} is \NP-hard, even when
every predefined path has at most three edges.
\end{proof}

\subsection{From \MINWP{} to Minimum-Weight Set Cover}
\label{sect-set-cover}

\paragraph*{Construction.}
Consider a variation graph \(G=(V,E,\pathSet)\). For each path \(h_i\in\pathSet\), we introduce a fake vertex \(t'_i\) and a fake edge \((\en(h_i),t'_i)\). Let \(T=\{(\en(h_i),t'_i)\mid h_i\in\pathSet\}\), and define \(\col((\en(h_i),t'_i))=\{i\}\). Recall that \(V_i=V(h_i)\). For every \(h_i\in\pathSet\) and every edge \(e\in E\cup T\) such that \(i\in\col(e)\), define \(W_i(e)=\{x\in V_i\mid x\text{ is a cardinal ancestor of }\be(e)\}\).

We construct an instance \((\mt U,\mt S,\mt W)\) of the minimum-weight set cover problem as follows. The universe is \(\mt U=\{(v,i)\mid h_i\in\pathSet,\ v\in V_i\}\). For every edge \(e\in E\cup T\), define \(S(e)=\{(v,i)\in\mt U\mid i\in\col(e),\ v\in W_i(e)\}\), and assign it weight \(\mt W(S(e))=|\col(e)|\). 
Finally, let \(\mt S=\{S(e)\mid e\in E\cup T\}\).

The construction can be carried out in polynomial time. Indeed, for every pair of vertices \(x,u\), whether \(x\) is a cardinal ancestor of
\(u\) can be determined by counting the directed paths from \(x\) to \(u\),
with the count capped at two, using dynamic programming on the DAG (Algorithm~\ref{algo:dp}).
Consequently, all sets \(W_i(e)\), and hence the entire set-cover instance,
can be constructed in polynomial time.

Observe that $(\en(h_i), i) \in S((\en(h_i), t'_i))$, but $(\en(h_i), i) \notin S(e)$ for any edge $e$ on $h_i$. 
Therefore, every fake edge $(\en(h_i), t'_i)$ must be included in any valid set cover of $(\mt {U}, {\cal S}, {\cal W})$, and each such edge contributes exactly one unit to the total weight.

Let $\mt P$ be a saturated pebbling of the original variation graph, and let $\{\mt G_1, \dots, \mt G_{|\pathSet|}\}$ be the induced guardian collection. 

\begin{lemma}\label{lem-reduction-ILP-1}
Let \(\mt E=\bigcup_i \mt G_i\). Then the sets associated with the guardian edges in $\mt E$ and the fake edges cover the entire universe \(\mt U\); that is, $\bigcup_{e\in \mt E\cup T} S(e)=\mt U.$
\end{lemma}

\begin{proof}
Fix \((v,i)\in\mt U\), and let $f_i=(\en(h_i),t'_i)$ be the fake edge placed
immediately after the last edge of \(h_i\). 
Traverse \(h_i\) from left to right and let \(e\) be the first edge of $\mt G_i\cup\{f_i\}$ whose tail is encountered at or after \(v\).

Let \(a=\be(h_i)\) if \(e\) is the first edge of \(\mt G_i\cup\{f_i\}\);
otherwise, let \(a\) be the head of the guardian edge immediately
preceding \(e\). By the unique path condition for \(\mt G_i\), there is
a unique path in \(G\) from \(a\) to \(\be(e)\).

The vertex \(v\) lies on the subpath of \(h_i\) from \(a\) to
\(\be(e)\). If there were two distinct paths from \(v\) to \(\be(e)\),
then prepending the subpath of \(h_i\) from \(a\) to \(v\) would give
two distinct paths from \(a\) to \(\be(e)\), a contradiction.
Therefore, \(v\) is a cardinal ancestor of \(\be(e)\), and hence $v\in W_i(e)$.
Consequently, \((v,i)\in S(e)\).

Thus every element of \(\mt U\) is covered by an edge of
\(\mt G_i\cup\{f_i\}\). Since every \(S(e)\) is a subset of
\(\mt U\), the union is exactly \(\mt U\).
\end{proof}

% Let $\mt E=\bigcup_i \mt G_i$ be the set of all edges appearing in the guardian collection induced by $\mt P$.
By Lemma~\ref{lem-reduction-ILP-1}, \(\mt   \{S(e)\mid e\in \mathcal E\cup T\}\) forms a set cover of the instance $(\mt {U}, {\cal S}, {\cal W})$.

\begin{lemma}\label{lem-reduction-ILP-2}
Let $E' \subseteq E \cup T$ be any set cover of the instance $(\mt {U}, {\cal S}, {\cal W})$.
For every $i \in [1 \dd |\pathSet|]$, the set 
$E'_i = \{e \mid e \in E' \setminus T,\; i \in \col(e)\}$
forms a guardian for the path $h_i$.
\end{lemma}

\begin{proof}
Fix \(i\in[1\dd|\pathSet|]\), and let \(\pi_i\) be any subpath of
\(h_i\) containing no edge of \(E'_i\). 
Let $a=\be(\pi_i)$ and $b=\en(\pi_i)$.
We prove that there is a unique path from \(a\) to \(b\).

Since \(E'\) is a set cover, the element \((a,i)\in\mt U\) is
covered by some selected set. Hence, there exists an edge
$e\in E'_i\cup \{(\en(h_i),t'_i)\}$
such that \(a\in W_i(e)\). Let \(u=\be(e)\). By the definition of
\(W_i(e)\), there is a unique path from \(a\) to \(u\).

Because \(i\in\col(e)\), either \(e\) is an edge of \(h_i\), or \(e\)
is the fake edge whose tail is \(\en(h_i)\). Moreover, \(u\) cannot
occur before \(a\) on \(h_i\), since \(a\) reaches \(u\).

Suppose that \(u\) occurred strictly before \(b\) on \(h_i\). Then
\(e\) would be an original edge of \(h_i\) lying on the subpath
\(\pi_i\). Since \(e\in E'_i\), this would contradict the assumption
that \(\pi_i\) contains no edge of \(E'_i\). Therefore, \(u=b\), or
\(u\) occurs after \(b\) on \(h_i\).

Let \(\rho\) be the subpath of \(h_i\) from \(b\) to \(u\), possibly
of length zero. If there were two distinct paths from \(a\) to \(b\),
then appending \(\rho\) to both of them would produce two distinct
paths from \(a\) to \(u\). This contradicts the fact that \(a\) is a
cardinal ancestor of \(u\).

Hence, there is a unique path from \(a\) to \(b\). In particular, the
subpaths between consecutive edges of \(E'_i\), as well as the initial
and final subpaths, satisfy the unique path condition. Therefore,
\(E'_i\) is a guardian for \(h_i\).
\end{proof}

Let \(C\subseteq E\cup T\) be a set of edges such that
\(\{S(e)\mid e\in C\}\) is a set cover of
\(\mt U\), and
define $\mt P'
=
\{(e,i)\mid e\in C\setminus T,\ i\in\col(e)\}.$
By Lemma~\ref{lem-reduction-ILP-2}, for every
\(i\in[1\dd|\pathSet|]\), the edges of \(C\setminus T\) carrying
color \(i\) form a guardian for \(h_i\). Therefore, \(\mt P'\) is a
pebbling. It is saturated by construction.
Every fake edge belongs to \(C\), and every fake edge has weight
one. 
Therefore, $\sum_{e\in C}\mt W(S(e))
=
\sum_{e\in C\setminus T}|\col(e)|+|T|
=
|\mt P'|+|\pathSet|.$

Let $\operatorname{OPT}_{\text{pebbling}}$ denote the size of a minimum saturated pebbling of $G$.
Since \(\mt P'\) is a saturated pebbling, $\operatorname{OPT}_{\text{pebbling}}+|\pathSet|
\le |\mt P'|+|\pathSet| =
\sum_{e\in C}\mt W(S(e)).$
As this holds for every set cover \(\{S(e)\mid e\in C\}\), it follows that $\operatorname{OPT}_{\text{pebbling}}+|\pathSet|
\le
\operatorname{OPT}_{\text{SC}}$, where $\operatorname{OPT}_{\text{SC}}$ denotes the optimum value of the weighted set-cover instance. 

Conversely, let \(\mt P\) be any saturated pebbling of \(G\), let
\(\{\mt G_1,\ldots,\mt G_{|\pathSet|}\}\) be its induced guardian
collection, and let \(\mt E=\bigcup_i\mt G_i\). By
Lemma~\ref{lem-reduction-ILP-1}, the family
\(\{S(e)\mid e\in\mt E\cup T\}\) is a set cover of \(\mt U\).
Since \(\mt P\) is saturated, the weight of this set cover is
$
\sum_{e\in\mt E\cup T}\mt W(S(e))
=
\sum_{e\in\mt E}|\col(e)|+|T|
=
|\mt P|+|\pathSet|.
$
Therefore, 
$\operatorname{OPT}_{\text{SC}}
\le \sum_{e\in\mt E\cup T}\mt W(S(e))=
|\mt P|+|\pathSet|.
$
As this holds for every saturated pebbling \(\mt P\), it follows that $\operatorname{OPT}_{\text{SC}}
\le
\operatorname{OPT}_{\text{pebbling}}+|\pathSet|.$

% By combining the two inequalities, it follows that the optimum value of the weighted set-cover instance is equal to the size of a minimum saturated pebbling of $G$, plus \(|\pathSet|\).
% Consequently, the pebbling \(\mt P^{*}\) obtained from a
% minimum-weight set cover is a minimum saturated pebbling of \(G\).
% Since the constructed set-cover instance can be produced in polynomial time, this proves Theorem~\ref{theorem-ILP}.

Combining the two inequalities yields $\operatorname{OPT}_{\text{SC}}
=
\operatorname{OPT}_{\text{pebbling}}+|\pathSet|.$
In particular, let \(C^*\subseteq E\cup T\) be such that
\(\{S(e)\mid e\in C^*\}\) is a minimum-weight set cover, and define $\mt P^*
=
\{(e,i)\mid e\in C^*\setminus T,\ i\in\col(e)\}.$
Then $|\mt P^*|
=
\operatorname{OPT}_{\text{SC}}-|\pathSet|
=
\operatorname{OPT}_{\text{pebbling}}$,
and hence \(\mt P^*\) is a minimum saturated pebbling of \(G\).
Since the constructed set-cover instance can be produced in polynomial
time, this proves Theorem~\ref{theorem-ILP}.

\begin{theorem}\label{theorem-ILP}
There is a polynomial-time reduction from \MINWP{} on variation graphs to Minimum-Weight Set Cover.
\end{theorem}

\subsection{The Integer Linear Programming Formulation}
\label{sect-ILP}

As a consequence of Theorem~\ref{theorem-ILP}, \MINWP{} can be solved using any exact algorithm for Minimum-Weight Set Cover. In particular, the constructed instance can be expressed using the standard $0$--$1$ integer linear programming formulation for set cover.

For every edge $e\in E\cup T$, define a binary variable $z_e$ that is
set to $1$ if and only if the set $S(e)$ is selected.  The minimum-weight
set-cover instance constructed above is represented by the following integer linear program (ILP):
\begin{equation}
	\label{eq:ilp-min-saturated-pebbling}
	\begin{aligned}
		\text{minimize}\qquad
		& \sum_{e\in E\cup T} \mt W(S(e))\times z_e
		= \sum_{e\in E\cup T} |\col(e)| \times z_e, \\
		\text{subject to}\qquad
		& \sum_{\substack{e\in E\cup T:\\(v,i)\in S(e)}} z_e \geq 1
		&& \text{for every $(v,i)\in\mt U$}, \\
		& z_e\in\{0,1\}
		&& \text{for every $e\in E\cup T$}.
	\end{aligned}
\end{equation}
The formulation has $|E\cup T|=|E|+|\mt {H}|$ binary variables and
$|\mathcal{U}|=\sum_{i=1}^{|\mathcal{H}|}|V_i|$  constraints.

For every feasible solution $X$, let $C_X=\{e\in E\cup T\mid z_e=1\}$ and $\mt {P}_X
=\{(e,i)\mid e\in C_X\setminus T,\ i\in\col(e)\}$.
The constraints guarantee that $\{S(e)\mid e\in C_X\}$ covers the universe $\mt {U}$; hence, by Lemma~\ref{lem-reduction-ILP-2}, $\mt {P}_X$ is a saturated pebbling.  Moreover, for every
$i\in[1..|\mt {H}|]$, the fake edge
$f_i=(\en(h_i),t'_i)$ must be selected in $X$, since the element
$(\en(h_i),i)$ belongs to $S(f_i)$ and to no other set.
Consequently, $z_{f_i}=1$ for every feasible solution, and $\sum_{e\in E\cup T}\mt {W}(S(e))\times z_e= |\mt {P}_X|+|\mt {H}|.$

Conversely, by Lemma~\ref{lem-reduction-ILP-1}, every saturated pebbling $\mt {P}$ induces a feasible solution by selecting all edges carrying at least one pebble, together with all fake edges in $T$.  Therefore, if $X^*=(z^*_e)_{e\in E\cup T}$ is an optimal solution to the ILP~\eqref{eq:ilp-min-saturated-pebbling}, then $\mt {P}^*
=\{(e,i)\mid e\in E,\ z^*_e=1,\ i\in\col(e)\}$ is a minimum saturated pebbling.
Moreover, $|\mt{P}^*|$ is exactly $|\mt{H}|$ smaller than the objective value of $X^*$.

\section{Queries on Pebbled Graphs}
\label{sect-query-pebbling}

In this section we will present the algorithms to answer edge and path queries using the data structures
of Lemma~\ref{lem-ds-pebbles} (that represents $\mt P$) and of Proposition~\ref{prop-graph-rep} (that represents $G$).
Specifically, the two queries we are interested in are:
\begin{description}
    \item[Edge query:] Given an edge $e \in E$, return $\col(e)$,
    \item[Path query:] Given $i \in [1 \dd |\mt \pathSet|]$, return $h_i$.
\end{description}

% We will first prove that an edge query can be answered in $\tilde{\mt O}(|E| + |\mt H| + |\mt P|)$ time and a path query in
% $\mt O(|E|)$ time. Then, we will show that if $\mt P$ is of minimum size, we can improve the edge query to $\tilde{O}(|E|+|\mt H|)$.
% In the remaining part of the section, we will see that if the pebbling is saturated, an edge
% query for $e \in E$ can be answered in $\tilde{\mt O}(|\kappa(e)| + |\partial(e)| + |\mt H|)$ time, where $\kappa(e)$
% and $\partial(e)$ are two sets of edges related to a certain subgraph of $G$.
% In both cases, due to space constraints, we will leave some special corner cases to the reader.

We first show that an edge query can be answered in $\tilde{\mt O}(|E|+|\mt H|+|\mt P|)$ time and a path query in $\mt O(|E|)$ time. We then prove that, when $\mt P$ has minimum size, the edge-query time can be improved to $\tilde{\mt O}(|E|+|\mt H|)$. Lastly, we show that when the pebbling is saturated, an edge query for $e\in E$ can be answered in $\tilde{\mt O}(|\kappa(e)|+|\partial(e)|+|\mt H|)$ time, where $\kappa(e)$ and $\partial(e)$ are two sets of edges associated with a certain subgraph of $G$ that depends on the query edge $e$ (\emph{cf.} Definition~\ref{def:boundary}).

% \pb{I would specify what kind of graph we consider ...more details}\bri{The graph is formally defined in its subsection. Added a pointer to the definition.}

% As for Section~\ref{sect-pebbling-to-compute}, we assume $G$ has been pre-computed in a way every path $h \in \mt H$
% starts at a common source $s$ end ends at a common sink $t$.

% \gyn{Why not use subsection to replace paragraph*?}
% \bri{Done!}
\subsection{Supporting Edge Queries}

Let $e = (u,v) \in E$ be the query edge.
For each $i\in \col(e)$, it follows that either $i\in \ms G(e) \cap \col(e)$ or $i \in \col(e) \setminus \ms G(e)$
Using query (b) of Lemma~\ref{lem-ds-pebbles} we can find $\ms G(e) \cap \col(e)$ in $\tilde{\mt {O}}(|\ms G(e)|)$ time.
Consider now the case $i \in \col(e) \setminus \ms G(e)$. The following lemma reduces the computation of the remaining $i \in \col(e) \setminus \ms G(e)$ 
to the existence of certain cardinal ancestors and descendants. To simplify the statement, we assume that for every $i$, the guardian 
$\mt G_i$ also contains the fictitious edges $(s_i, s_i)$ and $(t_i, t_i)$ as
first and last elements, respectively. The rank of $(s_i, s_i)$ is smaller
than every other edge from $h_i$, whereas the rank of $(t_i, t_i)$ is the largest.
Such edges will be considered non-traversable in the graph, so that they will not generate alternative paths and, thus, will not interfere
with the definition of cardinal ancestor and cardinal descendant.

First, given the query edge $e$, we distinguish two edges in $\mt G_i$. In the following definition, the rank of an edge
is defined with respect to the total order of $E$ as detailed in the preliminaries.
%\pb{Am  I wrong or $e=(u, v)? $ should be specified in the definition, since $e$ is given some rows above, too far}\bri{I state $u$ and $v$ only when I use them as names.}
\begin{definition}\label{def:prev-succ}
    Let $e=(u, v) \in E$ and $i \in [1 \dd |\mt H|]$ be such that $i \notin \ms G(e)$. We define:
    \begin{itemize}
        \item $\EPREV_i(e)$ is the edge $(x,y)\in\mt G_i$ of largest rank smaller than the rank of $e$ such that $y$ is a cardinal ancestor of $u$. If no such edge exists, then $\EPREV_i(e)=(s_i,s_i)$ provided that $s_i$ is a cardinal ancestor of $u$; otherwise, $\EPREV_i(e)=null$.
        \item $\ESUCC_i(e)$ is the edge $(x,y)\in\mt G_i$ of smallest rank
larger than the rank of $e$ such that $v$ is a cardinal ancestor of $x$. If no such edge exists, then
$\ESUCC_i(e)=(t_i,t_i)$ provided that 
$v$ is a cardinal ancestor of $t_i$; otherwise, $\ESUCC_i(e)=null$.
    \end{itemize}
\end{definition}

% \bri{The statement of Lemma~\ref{lem:card-colors} is the same as before but without self-loops.}

\begin{lemma}\label{lem:card-colors}
    Let $e = (u,v)$ and $i \in [1 \dd |\mt H|]$ such that $i \notin \ms G(e)$. Then, $i \in \col(e)$ if and only if there exist
    $e_u = (u'', u')$ and $e_v = (v', v'')$ such that all of the following hold:
    \begin{enumerate}
        \item $u'$ is a cardinal ancestor of $u$ with $(u'', u') \in \mt G_i$ or $(u'', u') = (s_i, s_i)$ and
        \item $v'$ is a cardinal descendant of $v$ with $(v', v'') \in \mt G_i$ or $(v', v'') = (t_i, t_i)$, and
        \item $e_u$ and $e_v$ are consecutive in $\mt G_i$, or $e_u = (s_i, s_i)$ and $e_v$ is the first element in $\mt G_i$, or $e_u$ is the
        last element in $\mt G_i$ and $e_v = (t_i, t_i)$, or $e_u = (s_i, s_i)$ and $e_v = (t_i, t_i)$ and $\mt G_i$ is empty.
    \end{enumerate}
    Moreover, if $i \in \col(e)$ then $e_u = \EPREV_i(e)$ and $e_v = \ESUCC_i(e)$.
\end{lemma}

% \begin{lemma}\label{lem:card-colors}
%     Let $e = (u,v)$ and $i \in [1 \dd |\mt H|]$ such that $i \notin \ms G(e)$. Then, $i \in \col(e)$ if and only if there exist
%     $e_u = (u'', u')$ and $e_v = (v', v'')$ in $\mt G_i$ such that all of the following hold:
%     \begin{enumerate}
%         \item $u'$ is a cardinal ancestor of $u$, and
%         \item $v'$ is a cardinal descendant of $v$, and
%         \item $e_u$ and $e_v$ are consecutive in $\mt G_i$.
%     \end{enumerate}
%     Moreover, if $i \in \col(e)$ then $e_u = \EPREV_i(e)$ and $e_v = \ESUCC_i(e)$.
% \end{lemma}
\begin{proof}
    ($\Rightarrow$) Let $i \in \col(e) \setminus \ms G(e)$. Since $e$ lies on path $h_i \in \mt H$ but is not pebbled by $i$, there must exist edges 
    $e_u = (u'', u')$ and $e_v = (v', v'')$ consecutive in $\mt G_i$ such that $e$ lies on the unique path going from $u'$ to $v'$. In particular,
    it must be that $u'$ is a cardinal ancestor of $u$ and $v'$ is a cardinal descendant of $v$, or the path from $u'$ to $v'$ would not be unique. Since $e_u$ and $e_v$ are consecutive in $\mt G_i$,
    and since $u'$ is an ancestor of $u$ and $v'$ is an ancestor of $v$, it must be
    $e_u = \EPREV_i(e)$ and $e_v = \ESUCC_i(e)$.
    
    ($\Leftarrow$) Let $i \notin \ms G(e)$, and $e_u$ and $e_v$ be as stated. We want to prove that $i \in \col(e)$. Since $u'$ is a cardinal
    ancestor of $u$ and $v'$ is a cardinal descendant of $v$, then there exists a unique path going from $u'$ to $u$ and going from $v$ to $v'$. 
    %\gyn{Here, probably I would suggest to mention that why $(u, v)$ cannot be a transitive edge in the beginning, otherwise the logic looks messed up.}
    Furthermore, there cannot exist another path going from $u$ to $v$ apart from the edge $(u,v)$, or we would have an edge $e' \in \mt G_i$ 
    such that $e_u < e' < e_v$, contradicting the fact that $e_u$ and $e_v$ are consecutive in $\mt G_i$.
    Therefore, edge $e$ must be on $h_i$ and, thus, $i \in \col(e)$.
\end{proof}

\paragraph*{Algorithm.} Lemma~\ref{lem:card-colors} above gives us a way to determine $\col(e)$: 
%\gyn{$\col(e) \setminus \ms G(e)$??? If I was not wrong, I think the algorithm below returns $\col(e) \setminus \ms G(e)$.}
% \bri{Fixed.}
\begin{enumerate}
    \item Determine the cardinal ancestors of $u$ and the cardinal descendants of $v$ by counting, for every $w \in V$, 
    how many paths in $G$ go from $w$ to $u$ and from $v$ to $w$. This can be done via DP in the same way as described
    in Section~\ref{sect-pebbling-to-compute} (this time two arrays suffice).
    %\gyn{Here, The DP is also mentioned. I think it is supposed to describe the details of the DB at least once somewhere in the paper, given that it has been references serveral time.}

    %\pb{I think we can keep the notion of }

    \item Compute $\EPREV_i(e)$ for every color $i \in [1 \dd |\mt H|]$, if it
        exists. This can be done by scanning $\ms G(u'',u')$ for every $u'' \in \inn(u')$ {(with the special case of $(s_i, s_i)$ as stated in Lemma~\ref{lem:card-colors})}, where $u'$ is a cardinal ancestor of $u$.
        Both operations are supported by the data structures that represent $\mt P$ and $G$ in $\tilde{\mt {O}}(|\ms G(u'', u')|)$
        and $\mt{O}(|\inn(u')|)$ for enumerating the respective sets. 
        
        Analogously, compute $\ESUCC_i(e)$ {by scanning $\ms G(v', v'')$ for
        every $v'' \in \out(v')$ (with the special case of $(t_i, t_i)$ as
        stated in Lemma~\ref{lem:card-colors}) where $v'$ is a cardinal
        descendant of $v$.} %\gyn{$\ESUCC_i(\cdot)$ defines over an edge. I saw several places where a vertex is used as the variable.}\bri{I can't find any other.}

    \item Add $\ms G(e)$ to the list $col$ to be returned. Then, for every color 
        $i \in [1 \dd |\mt H|]$, check whether $\EPREV_i(e)$ and $\ESUCC_i(e)$ are consecutive in 
        $\mt G_i$ and, if so, add $i$ to $col$. This check is supported by the data structure that represents $\mt P$ 
        in $\tilde{\mt{O}}(1)$ time.
\end{enumerate}

The pseudocode is given in Algorithm~\ref{algo:edge-query}.

\begin{algorithm}
\caption{Edge query for general pebbling.}\label{algo:edge-query}
\begin{algorithmic}[1]
    \nolinenumbers
    \Function{EdgeQuery}{$V, E, \mt P, e = (u,v)$}
        \State $to\_u[1 \dd |V|] \gets$ array initialized by 0\Comment{Step 1.}
        \State $\mt A \gets \emptyset$\Comment{Store cardinal ancestors.}
        \State $to\_u[u] \gets 1$
        \For{$x = u$ down to $1$}
            \If{$to\_u[x] = 1$}\Comment{New cardinal ancestor.}
                \State Prepend $x$ to $\mt A$
            \EndIf
            \For{$y \in \inn(x)$}
                \State $to\_u[y] \gets \min\{\,2, to\_u[y] + to\_u[x]\,\}$
            \EndFor
        \EndFor
        \State In a symmetric way compute cardinal descendants $\mt D$

        \State $E_u[1 \dd |\mt H|] \gets$ array initialized by $\mathtt{null}$\Comment{Step 2.}
        \For{$y \in \mt A$}
            \State {$E_u[i] \gets (s_i, s_i)$ for every color $i$ s.t. $y = s_i$}
            \For{$x \in \inn(y)$}\Comment{$\inn(y)$ is sorted increasing.}
                \For{$i \in \ms G(x, y)$}\Comment{Lemma~\ref{lem-ds-pebbles}, query (b).}
                    \State $E_u[i] \gets (x,y)$
                \EndFor
            \EndFor
        \EndFor
        \State In a symmetric way compute leftmost edges $E_v$ (enumerate $\out(\cdot)$ in reverse order)
        \State $col \gets \emptyset$\Comment{Step 3.}
        \State $col \gets col \cup \ms G(e)$\Comment{Lemma~\ref{lem-ds-pebbles}, query (b).}
        \For{$i = 1$ to $|\mt H|$}
            \State $j \gets |\mt G_i|+1$\Comment{Lemma~\ref{lem-ds-pebbles}, query (c).}
            \If{$E_u[i] \neq \mathtt{null}$}
                \If{$E_u[i] = (s_i, s_i)$}
                    \State $j \gets 0$
                \Else
                    \State $j \gets$ rank of $E_u[i]$ in $\mt G_i$\Comment{Lemma~\ref{lem-ds-pebbles}, query (e).}
                \EndIf
            \EndIf
            \State $nxt \gets 0$
            \If{$E_v[i] \neq \mathtt{null}$}
                \If{$E_v[i] = (t_i, t_i)$}
                    \State $nxt \gets |\mt G_i|+1$
                \Else
                    \State $nxt \gets$ rank of $E_v[i]$ in $\mt G_i$
                \EndIf
            \EndIf
            \If{$j+1 = nxt$} %\gyn{$j+1=nxt$?}
                    \State $col \gets col \cup \{i\}$
            \EndIf
        \EndFor
        \State\Return $col$
    \EndFunction
\end{algorithmic}
\end{algorithm}
%\pb{and the correctness...?}

\paragraph*{Correctness}

Step 1 is a simple variant of Algorithm~\ref{algo:dp}. Step 2 is a direct application
of Definition~\ref{def:prev-succ}. Step 3 follows from Lemma~\ref{lem:card-colors}.

\paragraph*{Complexity.}
Step 1 requires $\tilde{\mt O}(|V|+|E|) = \tilde{\mt O}(|E|)$ time. Step 2 requires, $\tilde{\mt O}(|E|+|\mt H|+|\mt P|)$ time
in the worst case, where $\mt O(|\mt P|)$ bounds the quantity $\sum_{(u'', u')} |\ms G(u'', u')|$.
Finally, Step 3 requires $\tilde{\mt O}(|\mt H|)$ time. In total, an edge query can be supported in $\tilde{\mt O}(|E|+|\mt H|+|\mt P|)$ time.
This proves Theorem~\ref{thm:edge-query-general}.

\begin{theorem}\label{thm:edge-query-general}
Let $\mathcal{G}=(V,E,\mt H)$ be a variation graph, and let $\mt P$ be a
pebbling of $\mathcal{G}$. In addition to the representation of the
underlying DAG, there exists a data structure using
$O(|E|+|\mt P|\log|\mt H|)$ bits that reports $\col(e)$ for any edge
$e\in E$ in $\widetilde{O}(|E|+|\mt H|+|\mt P|)$ time. 
\end{theorem}

\subsection{Improving Edge Queries for \MINSP{}.}

In case $\mt P$ is a minimum pebbling, Step 2 will run in $\tilde{O}(|E|+|\mt H|)$ time thanks to the following lemma.
\begin{lemma}\label{lem:min-peb-edgequery}
    Suppose $\mt P$ is a pebbling of minimum size, let $(u,v)$ be the query edge, and 
    $(x_1, y_1), (x_2, y_2), (x_3, y_3) \in E$ be three distinct edges. If $y_1$, $y_2$ and $y_3$
    are all cardinal ancestors of $u$, then $\ms G(x_1,y_1) \cap \ms G(x_2, y_2) \cap \ms G(x_3, y_3) = \emptyset$. The same holds if $x_1, x_2$ and $x_3$ are all cardinal descendants of $v$.
\end{lemma}

\begin{proof}
    For $k=1,2,3$ let $e_k = (x_k, y_k)$ and let $\ms G_k = \ms G(e_k)$. We prove the statement in case 
    $y_1, y_2$ and $y_3$ are all cardinal ancestors of $u$, since the other case is symmetric.

    Suppose, for the sake of contradiction, that there exists some color $i \in \ms G_1 \cap \ms G_2 \cap \ms G_3$,
    and assume without loss of generality that $y_1 < y_2 < y_3$. Indeed, if two were equal, say $y_1 = y_2$, then we would have
    that $h_i$ traverses $y_1$ twice and, hence, the graph would be cyclic. Since the $y_k$'s are all cardinal ancestors of $u$, 
    there exist unique paths $\pi_k$ 
    in $G$, for $k = 1,2,3$, going from $y_k$ to $u$. Furthermore, it must be that $\pi_3$ is a subpath of $\pi_2$ (or else $y_2$ would not be cardinal) and $\pi_2$ is a subpath of $\pi_1$ (or else $y_1$ would not 
    be cardinal), hence there exists a unique path in $G$ going from $y_1$ to $x_3$.
    Therefore, since the new guardian $\mt G'_i = \mt G_i \setminus \{e_2\}$ satisfies the unique path condition
    we have that $\mt P' = \mt P \setminus \{(e_2, i)\}$ is a valid pebbling for $G$, thus contradicting
    the assumption that $\mt P$ is of minimum size.
\end{proof}

By exploiting Lemma~\ref{lem:min-peb-edgequery} in the complexity analysis of Step 2 above, we notice that 
$\sum_{(u'',u')} |\ms G(u'', u')| \leq 2|\mt H|$, where $(u'', u')$ is a pebbled edge and $u'$ is a
cardinal ancestor of $u$.
Therefore, in the case of a minimum-size pebbling an edge query can be supported in $\tilde{\mt O}(|E|+|\mt H|)$ saving the extra $|\mt P|$ addend.
This proves Theorem~\ref{thm:edge-query-min-pebbling}.

\begin{theorem}\label{thm:edge-query-min-pebbling}
Let $\mathcal{G}=(V,E,\mt H)$ be a variation graph, and let $\mt P$ be a
minimum-size pebbling of $\mathcal{G}$. In addition to the
representation of the underlying DAG, there exists a data structure
using
$O(|E|+|\mt P|\log|\mt H|)$ bits that reports $\col(e)$ for any edge
$e\in E$ in $\widetilde{O}(|E|+|\mt H|)$
time.
\end{theorem}

\subsection{Supporting Path Queries}

Consider a path $h_i \in \mt H$ and its guardian $\mt G_i = (e_1, \dots, e_k)$, where $e_j = (u_j, v_j)$. 
We also let $v_0 := s_i$ and $u_{k+1} := t_i$. By the unique path 
condition, there exist unique paths in $G$ going from $v_j$ to $u_{j+1}$ for every $0 \leq j \leq k$. 
Our goal will be to detect those unique paths and return their concatenation.

\paragraph*{Algorithm.}
We give the pseudocode in Algorithm~\ref{algo:path-query}.
Remember that every vertex name corresponds to its rank in the sorted $V$, that is, the $u$-th vertex is $u$, and
that we denote by $e_j = (u_j, v_j)$ the $j$-th pebbled edge for path $h_i \in \mt H$, with the extensions of $v_0 = s_i$ and $e_{k+1} = (t_i, t_i)$, and
that $k = |\mt G_i|$.

Given the query color $i \in [1 \dd |\mt H|]$, we will compute an array of pairs $next[1 \dd |V|]$ such that, in the end, 
$next[u_j] = (v_j, j)$ and, for $u \neq u_j$, $next[u] = (v, j)$  if $j \in [1 \dd k+1]$  is the minimum integer such that there exists a path in $G$ going from 
$u$ to $u_j$, and $(u,v)$ is the first edge traversed on that path; if there is no path from $u$ to any $u_j$ we 
let $next[u] = (0,\infty)$. Note that it holds $next[t_i] = next[u_{k+1}] = (t_i, k+1)$.
We can compute $next$ by scanning $V$ from $t_i-1$ down to $s_i$: at the iteration of vertex $u$, if $u = u_j$ 
for some $j$, set $next[u] = (v_j, j)$; otherwise, determine any vertex $\bar{v} \in \out(u)$ minimizing 
the second element of pair $next[\bar{v}] = (w, j)$ and set $next[u] \gets (\bar{v}, j)$.
Note that if $u$ is on path $h_i$ then the minimizing $\bar{v} \in \out(u)$ is such that 
$next[\bar{v}] = (w,j) \neq (0, \infty)$, and there cannot exists a different $v \in \out(u)$ such that 
$next[v] = (\cdot, j)$.
%\gyn{The claim above does not seem to hold. See a possible counterexample in the figure below.}

% \begin{figure}[!h]
%     \centering
%     \includegraphics[scale=0.3]{../figs/counterex.png}
%     \caption{ \gyn{Please see whether the figure is a counter example: $\out(s)=\{a, t\}$, $next[a]=(b, 2)$, $next[t]=(t, 2)$. As a result, it is not clear how to set $next[s]$? Perhaps, setting $next[t]=(t, k+2)$ could solve your problem? I don't know; that is why we need a proof of correctness.}\bri{See above.}
%     }\label{fig:counter}
%         % $\mt P = \left\{(e_2, 1), (e_7,1), (e_{10},1), (e_1, 2), (e_5, 2), (e_{10}, 2), (e_1, 3), (e_4, 3), (e_6, 3), (e_8, 3), (e_{11}, 3)
%         % (e_3, 4), (e_9, 4), (e_{11}, 4)\right\}$
% \end{figure}

Once array $\rm{next}$ is computed, we can reconstruct $h_i$ as follows, using variable $H$ to store the partially
reconstructed path (initialized to the empty sequence) and variable $curr$ (initially set to $curr = s_i$). 
At each iteration, if $curr = t_i$ exit the procedure and return $H$; otherwise, let $next[curr] = (w, j)$ 
(where $j \neq \infty$ necessarily), and update $H \gets (H, (curr, w))$ and $curr \gets w$.

\begin{algorithm}
\caption{Algorithm for path query.}\label{algo:path-query}
\begin{algorithmic}[1]
    \nolinenumbers
    \Function{PathQuery}{$V, E, \mt P, i$}
        \State $next[1 \dd |V|] \gets$ array initialized by $(0, \infty)$
        \State $j \gets k \gets |\mt G_i|$\Comment{Lemma~\ref{lem-ds-pebbles}, query (c).}
        \State $next[t_i] \gets (t_i, k+1)$
        \For{$u = t_i-1$ down to $s_i$}
            \If{$j > 0$ and $u = u_j$}\Comment{Lemma~\ref{lem-ds-pebbles}, query (d)}
                \State $next[u] \gets (v_j, j)$ %\gyn{Perhaps I am missing something. Consider the simple path
% $s\rightarrow a\rightarrow t$ with the empty guardian, so that
% $k=j=0$. When the loop reaches $s_i=s$, the condition $u=s_i$ holds,
% and the algorithm sets $next[s_i]=(v_0,0)=(s_i,0).$
% Would this cause the reconstruction loop to remain forever at $s_i$?}
% \bri{Indeed ``$u = s_i$'' didn't even make any sense for the logic of the algorithm, I don't know what I was thinking about.}
                \State $j \gets j-1$ 
            \Else
                \For{$v \in \out(u)$}
                    \If{$next[u].second \geq next[v].second$} 
                    %\gyn{Can both $seconds$ are infinity?}
                    % \bri{Why is this a problem?}
                        \State $next[u] \gets (v, next[v].second)$
                    \EndIf
                \EndFor
            \EndIf
        \EndFor\label{algo:path-query:endfor}
        \State $curr \gets s_i$
        \State $H \gets \emptyset$
        \While{$curr \neq t_i$}
            \State Append $(curr, next[curr].first)$ to $H$
            \State $curr \gets next[curr].first$
        \EndWhile
        \State\Return $H$
    \EndFunction
\end{algorithmic}
\end{algorithm}

%\gyn{A description of proof of the correctness would be appreciated.}
% \bri{Done.}

\paragraph*{Correctness.}

It is sufficient to prove the following lemma, since the construction of the list $H$
is straightforward.
\begin{lemma}\label{lem:corr-path-query}
    After line~\ref{algo:path-query:endfor} of Algorithm~\ref{algo:path-query},
    for every $u \in V$ the following hold:
    \begin{enumerate}
        \item If $u = u_j$ for some $1 \leq j \leq k+1$, then $next[u] = (v_j, j)$, otherwise

        \item If $next[u].second = \infty$, then there is no path in $G$ going from $u$ to any $u_j$, otherwise

        \item If $next[u] = (v, j)$, then $j$ is the minimum integer such that there exists a path in $G$ 
            going from $u$ to $u_j$, and $v$ is on that path. Moreover, if $u$ is on path $h_i$, 
            $v$ is the unique minimizer of $next[w].second$ for $w \in \out(u)$.
    \end{enumerate}
\end{lemma}
\begin{proof}
    Since every vertex $u$ is the $u$-th vertex in $V$, we treat vertices as integers,
    and we denote by $x_n$ the vertex $t_i-n$, for every $t_i - |V| \leq n < t_i$. In
    particular, when $n < 0$ we have $t_i < x_n \leq |V|$.
    
    We prove the claim for vertex $x_n$ by strong induction on $t_i - |V| \leq n < t_i$.

    (Base) For $n < 0$ the claim trivially holds, since $x_n > t_i$ and there are no paths
        from $x_n$ to any $u_j$, and $next[x_n] = (0, \infty)$ is never updated. For
        $n = 0$ we have $x_n = t_i = u_{k+1} = v_{k+1}$, and it holds 
        $next[u_{k+1}] = (v_{k+1}, k+1)$.
        
    (Step) Assume the claim holds for all $k \leq n$, we prove it for $n+1$.
        \begin{enumerate}
            \item If $x_{n+1} = u_j$ for some $j > 0$, then the algorithm 
                correctly sets $next[x_{n+1}] = (v_j, j)$.

            \item If $next[x_{n+1}].second = \infty$, then for every $w \in \out(x_{n+1})$ (if any) it holds
                $next[w].second = \infty$. Since $w = x_k$ for some $k \leq n$, by induction
                hypothesis, there is no path going
                from $w$ to any $u_j$, thus there is no path from $x_{n+1}$ to any $u_j$ either.
                Since the entry for $x_{n+1}$ is never updated, in the end it will correctly hold 
                $next[x_{n+1}] = (0, \infty)$.

            \item If $next[x_{n+1}] = (v, j)$ and $j \neq \infty$, then by the algorithm it holds 
                $j = \min\{next[w].second : w \in \out(x_{n+1})\}$, and by induction hypothesis $j$ 
                is also the minimum integer for which there exists a path in $G$ going from $v$ to $u_j$. 
                Thus, there exists a path from $x_{n+1}$ to $u_j$ via $v$. If there was a path from $x_{n+1}$ 
                to $u_{j'}$ for some $j' < j$, it would traverse some $v' \in \out(x_{n+1})$ and
                thus $j$ would not be chosen as the minimum. Finally, assume $x_{n+1}$ is on path $h_i$ and
                suppose, for the sake of contradiction, that there exists $v' \in \out(x_{n+1})$ distinct from $v$
                such that $next[v'].second = j$. Then, by induction hypothesis we have that there exists a
                path going from $v'$ to $u_j$ and, thus, there would be two distinct paths going from
                $x_{n+1}$ to $u_j$ and, in particular, two distinct paths going from $v_{j-1}$ to $u_j$, 
                contradicting the unique path condition of $\mt G_i$.
                Hence $next[x_{n+1}] = (v, j)$ is correct.
        \end{enumerate}
\end{proof}

\paragraph*{Complexity.} Every vertex and every edge of $G$ is visited once. For each iteration $u$, enumerating $\out(u)$ takes $\mt{O}(|\out(u)|)$ time. Testing whether
$u = u_j$ takes constant time (using Lemma~\ref{lem-ds-pebbles}, query (d)). Hence, in total the path query is supported in $\mt{O}(|V|+|E|) \subseteq \mt{O}(|E|)$ time.
This gives Theorem~\ref{thm:path-query}.

%\gdv{Notice that the results in this section hold for all guardians of $\mt H$, even if the guardian does not have minimum size. At the same time, the time and space complexities depend on the size of the guardian, which justifies our focus on smallest guardians. But even approximate solutions to the  \MINSP{} result if efficient queries.}\bri{Agree
%with this comment. We should put it somewhere. Maybe at the beginning of the section?}

\begin{theorem}\label{thm:path-query}
Let $\mathcal{G}=(V,E,\mt H)$ be a variation graph, and let $\mt P$ be a
pebbling of $\mathcal{G}$. In addition to the representation of the
underlying DAG, there exists a data structure using
$O(|E|+|\mt P|\log|\mt H|)$ bits that, given a color
$i\in[1\dd|\mt H|]$, sequentially reports all edges of the path $h_i$
in $O(|E|)$ time.
\end{theorem}

\subsection{Exploiting Saturation}

We have seen that for an unrestricted pebbling we incur into an $\tilde{\mt O}(|E|+|\mt H|+|\mt P|)$ running time
to answer an edge query, and this time can be reduced to $\tilde{\mt O}(|E|+|\mt H|)$ in case $\mt P$ is of minimum size. 
Thus, in the worst case, we still have to look at every edge in the graph; this is due to the fact that multiple pebbled edges of different colors may
lie on a same path in $G$. For this reason, when (an edge incident to) a cardinal ancestor $u'$ (or descendant $v'$)
is pebbled by some color $i$, we must continue the search for other ancestors/descendants beyond $u'$ and $v'$.

In this final part, we present an algorithm that exploits the saturation property to reduce the running time to $\tilde{\mt O}(|\kappa(e)|+|\partial(e)|+|\mt H|)$, where $\kappa(e)$ and $\partial(e)$ are determined by a subgraph of $G$ induced by $\mt P$ and the query edge $e$. Intuitively, when exploring the descendants of $v$ or the ancestors of $u$, the search along a path can stop as soon as a pebbled edge is encountered. All the pebbled edges found in this way
will be sufficient to determine $\col(e)$.

% In this final part we will present a heuristic\gyn{I would not call it ``heuristic''}\bri{Change it as you prefer.} that exploits the saturation property in order
% to reduce the running time to $\tilde{\mt O}(|\kappa(e)|+|\partial(e)|+|\mt H|)$, where $\kappa(e)$ and $\partial(e)$
% depend on a particular subgraph of $G$ induced by $\mt P$ and $e$. Intuitively, when
% exploring the descendants of $v$ (or the ancestors of $u$) along some path and a pebbled edge is encountered, the
% exploration along that path can be stopped. All the pebbled edges found in this way
% will be sufficient to determine $\col(e)$.

In what follows, we will assume that the query edge $e$ is not pebbled. Indeed, since
the pebbling is saturated, in case $e$ is pebbled it holds $\col(e) = \ms G(e)$, and
the query can be answered directly using Lemma~\ref{lem-ds-pebbles} query (b).
Moreover, as already done for Definition~\ref{def:prev-succ}, in order to simplify
the exposition we assume to have in $G$ the fictitious edges $(s_i, s_i)$ and $(t_i, t_i)$
that are also pebbled by $i$ (hence, they are also elements of $\mt G_i$).
%\pb{In my opinion we can keep these edges but I would explicitly call them loop edges, but I do not like the idea that they are also part of the guardian list, as the definition  consider $s_i$  ad $t_i$...another solution would be that we modify the notion of guardian from the beginning with these loop edges ?}

We start by formally defining the subgraph induced by the query edge.
\begin{definition}\label{def:boundary}
    Let $e = (u, v) \in E$ with $\ms G(e) = \emptyset$. We denote by $G(e)$ the weakly connected component of $G$ that contains $e$ obtained by removing all pebbled edges in $G$. Furthermore, we say that
    a pebbled edge $(x,y)$ is a \emph{boundary edge for $G(e)$} if $x$ is in $G(e)$ and it is a descendant of $v$,
    or $y$ is in $G(e)$ and it is an ancestor of $u$. The collections of edges of $G(e)$ and of boundary edges for $G(e)$ are denoted by 
    $\kappa(e)$ and $\partial(e)$, respectively.
%\gyn{I was planning to help address Gianluca's first comment in the Introduction. Then I noticed that the definition of $\partial(e)$ had been changed. In fact, defining $\partial(e)$ as the set of pebbled edges with at least one endpoint in $G(e)$ is harmless and would make Gianluca's comment easy to address. In any case, if you prefer the new definition, that is also fine with me.}\bri{For me it is ok in either ways, choose what you prefer.}
\end{definition}

Recall Definition~\ref{def:prev-succ} of $\EPREV_i(\cdot)$ and $\ESUCC_i(\cdot)$. The following
is a simple fact that follows from the saturation property.

% \gyn{Imagine that the input graph consists of the single path
% $s\rightarrow a\rightarrow b\rightarrow t$. No pebbles are needed.
% For the query edge $(a,b)$, Definition~17 gives
% $\EPREV_i((a,b))=(s_i,s_i)$. However, $(s_i,s_i)$ is a fictitious edge
% and therefore does not belong to $\partial((a,b))$, does it?} \bri{I've explained more precisely
% what is a fictitious edge.}
% \gyn{Note that $(s_i, s_i)\notin E$, so $(s_i, s_i)\notin \partial(e)$?}\bri{``Moreover, as already done for Definition~\ref{def:prev-succ}, in order to simplify
% the exposition we assume to have in $G$ the fictitious edges $(s_i, s_i)$ and $(t_i, t_i)$
% that are also pebbled by $i$ (hence, they are also elements of $\mt G_i$).'' I removed the membership to $E$ since it was very very confusing.}
\begin{fact}\label{fact:prev-succ}
Let $e=(u, v)$ be unpebbled and $i\in \col(e)$. 
Then
$\EPREV_i(e)\neq null$ and
$\ESUCC_i(e)\neq null$.
Moreover, if $\EPREV_i(e)=(x, y)$ (resp. $\EPREV_i(e)=(s_i, s_i)$), then there exists an unpebbled path from $y$ (resp. $s_i$) to $u$.
Symmetrically, if $\ESUCC_i(e)=(x, y)$ (resp. $\ESUCC_i(e)=(t_i, t_i)$), then there exists an unpebbled path from $v$ to $x$ (resp. $t_i$).
\end{fact}

% Fact~\ref{fact:prev-succ} together with Lemma~\ref{lem:card-colors} ensure that
% to determine $\col(e)$ it is sufficient to focus on the graph $G^+(e) = G(e)\cup\partial(e)$ and,
% in particular, on the edges in $\partial(e)$, {plus the pairs
% $(s_i, s_i), (t_i, t_i)$ that act as special cases.}
Fact~\ref{fact:prev-succ}, together with Lemma~\ref{lem:card-colors}, implies that, to determine $\col(e)$ for an unpebbled edge $e$, it suffices to explore $G(e)$ along unpebbled edges and consider the pebbled edges encountered at its boundary. More precisely, the backward traversal from $u$ identifies the possible values of $\EPREV_i(e)$, while the forward traversal from $v$ identifies the possible values of $\ESUCC_i(e)$.

% \begin{figure}[!t]
%     \centering
% \includegraphics[scale=0.4]{../figs/colored_pebbling.png}
%     \caption{In this example, each edge itself is an independent haplotype path. The crosses indicate the edges that are pebbled. Note that it is a saturated pebbling, but might not be an optimal one. The connected component $G((u, v))$ consists of the edges $(1,2), (2, 4)$ and $(3,4)$. The traversal algorithm might not be able to access the pebbled edge $(0,1)$, but $(0,1)$ is a boundary edge for that component.\bri{Edge $(0,1)$
%     is not a boundary edge of this component, since neither $1$ is an ancestor of $u=3$ nor $0$ is a descendant of $v=4$.}
%     }\label{fig:counterexp}
% \end{figure}

\paragraph*{Algorithm.}

Let $e = (u,v) \in E$ be the query edge. If $e$ is pebbled, we can return $\ms G(e)$ using
Lemma~\ref{lem-ds-pebbles}, query (b). Otherwise, in light of the above discussion, we will
find $\EPREV_i(e)$ and $\ESUCC_i(e)$ for every $i$ and check if they are consecutive in $\mt G_i$.
The algorithm is described as follows:

% One possible way is to use an algorithm similar to Algorithm~\ref{algo:edge-query} where only
% the cardinal ancestors and descendants in $G(e)$ are computed. This would require to maintain
% $G^+(e)$ in memory. Instead, we present a different algorithm specifically tailored to use the
% saturation property:\\
\begin{enumerate}
    \item Construct two sets $A$ and $D$ by performing two directed traversals that use only unpebbled edges.\\
    
    To construct $A$, start from $u$ and traverse the graph backwards. Whenever a vertex $y$ is visited, inspect all its incoming edges. If an incoming edge $(x, y)$ is unpebbled, continue the traversal from $x$. If $(x, y)$ is pebbled, add it to $A$, but do not cross that edge. Moreover, for every visited vertex $y$, add the fictitious edge $(s_i, s_i)$ for every color $i$ such that $s_i=y$.
    Thus, $A$ contains the pebbled edges $(x, y)$ for which there exists an unpebbled path from $y$ to $u$, as well as the fictitious edge $(s_i, s_i)$ for which there exists an unpebbled path from $s_i$ to $u$.\\

    Symmetrically, construct $D$ by starting from $v$ and traversing the graph forward. Whenever a vertex $x$ is visited, inspect all its outgoing edges. If an outgoing edge $(x, y)$ is unpebbled, continue the traversal from $y$. If $(x, y)$ is pebbled, add it to $D$, but do not cross that edge. Moreover, for every visited $x$, add the fictitious edge $(t_i, t_i)$ for every color $i$ such that $t_i=x$. Thus, $D$ contains the pebbled edge $(x, y)$ for which there exists an unpebbled path from $v$ to $x$, as well as the fictitious edge $(t_i, t_i)$ for which there exists an unpebbled path from $v$ to $t_i$.\\
    % Collect the set $B = \partial(u,v) \cup \{(s_i, s_i) : \text{$s_i$ in $G(e)$ is ancestor of $u$}\} \cup \{(t_i, t_i) : \text{$t_i$ in $G(e)$ is descendant of $v$}\}$. This can be done
    %     by performing a traversal \gyn{If I was not wrong, you need two traversals. Please elaborate more on at least one of the two traversals such as where the traversal begins, the directions of the traversals, i.e., forward or backward, and how the source and sink sentinels $(s_i, s_i)$ or $(t_i, t_i)$ are collected, etc.}\bri{You can do both directions with a unique traversal, you just need a flag to say "go back" or "go forth". Think about a BFS with two sources ($u$ and $v$) and the elements in the queue have this additional flag. Am I missing something?} of $G$ that does not continue the exploration past a pebbled edge. \gyn{I think the algorithm might not collect the entire set $B$, to be precise. Instead, it should collect the set $\{(x, y)\mid \text{$(x, y)$ is pebbled and there is an unpebbled path from $y$ to $u$.}\}$. See Figure~\ref{fig:counterexp} for details.}\bri{Edge $(0,1)$ in Figure~\ref{fig:counterexp} does not satisfy the definition of boundary edge.}
    \item Initialize the arrays $Prev[1..|\mt H|]$ and $Succ[1..|\mt H|]$ with null values. \\

    Sort the elements of $A$ in decreasing order of their edge ranks, treating every fictitious edge $(s_i, s_i)$ as smaller than every ordinary edge. Scan the elements of $A$ in this order.
    For each ordinary pebbled edge $e'\in A$, enumerate the color $\mt G(e')$.
    Whenever a color $i$ is encountered, set $Prev[i]$ to $e'$ if $Prev[i]$ is still null. If $Prev[i]$ has already been assigned, then stop enumerating the remaining colors of $\mt G(e')$ and continue with the next element of $A$. When the element being processed is $(s_i, s_i)$, process its only associated color $i$ and assign $(s_i, s_i)$ to $Prev[i]$ if $Prev[i]$ is still null.\\
    
    Symmetrically, sort the elements of $D$ in increasing order of their edge ranks, treating every fictitious edge $(t_i, t_i)$ as larger than every ordinary edge. Scan the elements of $D$ in this order.
    For each ordinary pebbled edge $e'\in D$, enumerate the color $\mt G(e')$.
    Whenever a color $i$ is encountered, set $Succ[i]$ to $e'$ if $Succ[i]$ is still null. If $Succ[i]$ has already been assigned, then stop enumerating the remaining colors of $\mt G(e')$ and continue with the next element of $D$.
    When processing a fictitious edge $(t_i, t_i)$, simply set $Succ[i]=(t_i, t_i)$ if $Succ[i]$ is stil null.\\

    \item Initialize $\col(e)$ as an empty set.
    For each color $i$, consider two edges stored in $Prev[i]$ and $Succ[i]$.
    If either entry is null, then skip color $i$.
    Otherwise, treat the fictitious edge $(s_i, s_i)$ as appearing immediately before the first edge of $\mt G_i$, and treat $(t_i, t_i)$ as appearing immediately after the last edge of $\mt G_i$.
    Check whether $Prev[i]$ and $Succ[i]$ are consecutive in this augmented order, meaning that no edge of $\mt G_i$ lies between them.
    This also covers the case in which $Prev[i]=(s_i, s_i)$, $Succ[i]=(t_i, t_i)$, or both occur when $\mt G_i$ is empty.
    If they are consecutive, add $i$ to $\col(e)$.
    After all colors have been processed, return $\col(e)$.
\end{enumerate}

\paragraph*{Correctness.}

We first observe the following property.
\begin{lemma}\label{lem:prev-all}
    Let $e=(u, v)$ be not pebbled and $i \in \col(e)$. Let $e_u = \EPREV_i(e)$ and $e_v = \ESUCC_i(e)$. Then the following hold:
    \begin{enumerate}
        \item If $e_u = (s_i, s_i)$, then $e_u = \EPREV_j(e)$
            for every $j$ such that $s_i = s_j$; otherwise, $e_u = \EPREV_j(e)$ for every $j \in \ms G(e_u)$.
        
        \item If $e_v = (t_i, t_i)$, then $e_v = \ESUCC_j(e)$
            for every $j$ such that $t_i = t_j$; otherwise, $e_v = \ESUCC_j(e)$ for every $j \in \ms G(e_v)$.
    \end{enumerate}
\end{lemma}

\begin{proof}
    We prove only (1), since the proof of (2) is symmetric.

    Let $r=u'$ if $e_u=(u'',u')\in\mt G_i$, and let $r=s_i$ if
    $e_u=(s_i,s_i)$. By Definition~\ref{def:prev-succ}, $r$ is a
    cardinal ancestor of $u$.
    Consider a color $j\in\ms G(e_u)$ when $e_u$ is an ordinary edge,
    or a color $j$ such that $s_j=s_i$ when $e_u=(s_i,s_i)$.
    Suppose, for the sake of contradiction, that
    $\EPREV_j(e)=e'\neq e_u$.

    As $\EPREV_j(e)=e'$ and $e'\neq e_u$,
    the edge $e'$ must occur after $e_u$ in $h_j$. Let $e'=(x,y)$. By the
    definition of $\EPREV_j(e)$, $y$ is a cardinal ancestor of $u$.
    Therefore, the subpath of $h_j$ from $r$ to $y$, followed by the
    unique path from $y$ to $u$, forms a path from $r$ to $u$ that
    traverses $e'$.

    Since $r$ is a cardinal ancestor of $u$, this path is the unique
    path from $r$ to $u$. In particular, it coincides with the subpath
    of $h_i$ from $r$ to $u$. Hence, $h_i$ traverses $e'$. Since the
    pebbling is saturated, this implies $i\in\ms G(e')$.
    Moreover, $y$ is a cardinal ancestor of $u$ and
    $e_u<e'<e$. Thus, $e'$ satisfies the conditions in the definition
    of $\EPREV_i(e)$ and has rank larger than that of $e_u$, contradicting
    the choice of $e_u=\EPREV_i(e)$. Hence, it has to be $\EPREV_j(e)=e_u$.
\end{proof}

\pb{give an intuition of the meaning of the previous result}
% \begin{proof}
%     Let $e = (u,v)$ be not pebbled and $i \in \col(e)$. We prove only (1), since the 
%     proof of (2) is analogous.
%     For every $w \in V$ we extend $\ms G(\cdot)$ as 
%     $\ms G(w, w) = \{j \in [1 \dd \mt |H|] : w = s_j\}$.
    
%     Suppose, for the sake of contradiction, that $e_u = (u'',u') = \EPREV_i(e)$ and 
%     $e' = (x, y) = \EPREV_j(e)$ with $e_u \neq e'$ for some $j \in \ms G(e_u)\setminus\{i\}$. 
%     By Definition~\ref{def:prev-succ}, it follows $j \in \ms G(e_u) \cap \ms G(e')$ and $e_u < e'$. 
%     Thus, $h_j$ traverses both $e_u$ and $e'$ in this order. Furthermore, $i \notin \ms G(e') = \col(e')$, 
%     or $e_u$ would not be $\EPREV_i(e)$. Since $i \in \col(e)$, path $h_i$ traverses both $e_u$ and $e$ in 
%     this order, and does not traverse $e'$. Therefore, there exists two paths in $G$ going from $u'$ to $u$:
%     one is the subpath of $h_i$ traversing both $e_u$ and $e$ but not $e'$, the other is the concatenation between the subpath of $h_j$
%     traversing both $e_u$ and $e'$, and any path going from $y$ to $u$ (which must exist since $y$ is an ancestor of $u$). 
% Thus, either $\mt G_i$ does not satisfy the unique path
%     condition, or there exists another edge $e'' \in \partial(e)$ such that $i \in \ms G(e'')$ and $e_u < e''$,
%     contradicting the fact that $e_u = \EPREV_i(e)$.
% \end{proof}

For the correctness of the algorithm, we prove that Step~2 correctly identifies $\EPREV_i(e)$ and $\ESUCC_i(e)$ of every color $i$ whose path traverses the query edge $e$.

\begin{lemma}\label{lem:edge-sat-corr}
    At the end of Step 2, for every $i \in \col(e)$ it holds $Prev[i] = \EPREV_i(e)$
    and $Succ[i] = \ESUCC_i(e)$.
\end{lemma}

\begin{proof}
    Let $i\in\col(e)$ and let $e_u=\EPREV_i(e)$.
    By Fact~\ref{fact:prev-succ} and the construction of $A$, the
    element $e_u$ belongs to $A$ and is associated with color $i$.
    
    Suppose, for the sake of contradiction, that $Prev[i]\neq e_u$.
    Consider the iteration in which $e_u$ is processed. Since $e_u$ is associated with color $i$, if the enumeration reached $i$ while $Prev[i]=null$, the algorithm would set $Prev[i]=e_u$.
    Therefore, either $Prev[i]$ was already non-null when $i$ was reached, or the enumeration stopped earlier at some color $j$ whose entry was already non-null. In the former case, let $j=i$; in the latter case, let $j$ be the color at which the enumeration stopped. In both cases, $j$ is associated with $e_u$ and $Prev[j]=e'\neq e_u$ for some element $e'$ processed before $e_u$.

    Since $e'$ was processed before $e_u$ in the decreasing-rank scan
    of $A$, we have $e_u<e'$. By Lemma~\ref{lem:prev-all}, we have $e_u=\EPREV_j(e)$.
    Let $r$ be the head of $e_u$ if $e_u$ is an ordinary edge, and
    let $r=s_i=s_j$ otherwise. By
    Definition~\ref{def:prev-succ}, $r$ is a cardinal ancestor of $u$.

    Write $e'=(a,b)$. Since $e'\in A$, there exists an unpebbled path
    $\pi$ from $b$ to $u$. Moreover, $Prev[j]=e'$ implies that $e'$
    is associated with color $j$. Therefore, $h_j$ traverses both
    $e_u$ and $e'$ in this order. Let $\rho$ be the subpath of $h_j$
    from $r$ to $b$. The concatenation $\rho\cdot\pi$ is a path from
    $r$ to $u$.

    Since $r$ is a cardinal ancestor of $u$, this path is unique.
    Consequently, $b$ is also a cardinal ancestor of $u$.
    Hence, $e'$ is an edge of $\mt G_j$ whose head is a cardinal
    ancestor of $u$. Furthermore, by the construction of $A$ and the
    topological order, $e_u<e'<e$.
    Thus, the head of $e'$ is a cardinal ancestor of $u$ with
    rank larger than that of $e_u$, contradicting
    $e_u=\EPREV_j(e)$.

    Therefore, $Prev[i]=\EPREV_i(e)$. Symmetrically, we obtain
    $Succ[i]=\ESUCC_i(e)$.
\end{proof}

We can now prove the correctness of the algorithm.
If the query edge $e=(u, v)$ is pebbled, then, by saturation, $\col(e)=\mt G(e)$, so the answer returned directly by the algorithm is correct. Hence, assume that $e$ is unpebbled.

Consider any color $i\in \col(e)$.
We show that $i$ is added to the output in Step~3.
By Fact~\ref{fact:prev-succ} and the construction of $A$ and $D$, both
$\EPREV_i(e)$ and $\ESUCC_i(e)$ are included in $A$ and $D$, respectively. By Lemma~\ref{lem:edge-sat-corr}, we have ${Prev}[i]=\EPREV_i(e)$ and ${Succ}[i]=\ESUCC_i(e)$ at the end of Step~2.
By Lemma~\ref{lem:card-colors}, these two edges are consecutive in the augmented guardian
of $h_i$. Therefore, Step~3 adds $i$ to the output.

Let $i$ be any color added to the output in Step~3.
Next, we prove that $i\in \col(e)$.
Let $e_u=Prev[i]$ and $e_v=Succ[i]$. Both entries are non-null and are consecutive in the augmented guardian of $h_i$, by the condition in Step~3.

Let $r$ be the head of $e_u$ if $e_u$ is an ordinary edge, and let
$r=s_i$ if $e_u=(s_i,s_i)$. Since $e_u\in A$, by the construction of
$A$ there exists an unpebbled directed path $\pi_u$ from $r$ to $u$.
Similarly, let $z$ be the tail of $e_v$ if $e_v$ is an ordinary edge,
and let $z=t_i$ if $e_v=(t_i,t_i)$. Since $e_v\in D$, by the construction
of $D$ there exists an unpebbled directed path $\pi_v$ from $v$ to $z$.

As a result, the concatenation $\pi_u\cdot e\cdot\pi_v$ is a directed path from $r$ to $z$ that traverses the query edge $e$. On the other hand, since $e_u$ and $e_v$ are consecutive in the augmented guardian of $h_i$, the unique path condition guarantees that there exists a unique directed path from $r$ to $z$. The corresponding
subpath of $h_i$ is such a path. It must therefore coincide with $\pi_u\cdot e\cdot\pi_v$, and hence $e$ belongs to $h_i$. Thus, $i\in\col(e)$.

Therefore, the set returned by the algorithm is exactly
$\col(e)$.

\paragraph*{Complexity.}

In Step 1, the graph traversal and the construction of $B$ cost 
$\tilde{\mt O}(|\kappa(e)| + |\partial(e)| + |\mt H|)$ time. In Step 2, the total time spent for array
$Prev$ is proportional to the number $r$ of reads and $w$ of writes. Note that there is at most
one write per color (from ${null}$ to some value) and at most $w+|\mt A|$ reads
in total (one before every write and possibly an extra one per edge in $\mt A$ to stop the enumeration),
thus a total of $\tilde{\mt O}(|\mt A| + |\mt H|)$ time.
In the same way, the time spent for array $Succ$ is $\tilde{\mt O}(|\mt D| + |\mt H|)$. In total,
the time for Step 2 is $\tilde{\mt O}(|\partial(e)| + |\mt H|)$. Finally, Step 3 requires
$\tilde{\mt O}(|\mt H|)$ time.

Therefore, the entire algorithm runs in $\tilde{\mt O}(|\kappa(e)|+|\partial(e)|+|\mt H|)$.
The working space requires $O\bigl((|\kappa(e)|+|\partial(e)|+|\mt H|)\log |V|\bigr)$ bits.
This proves Theorem~\ref{thm:saturated-edge-query}.

\begin{theorem}\label{thm:saturated-edge-query}
Let $\mathcal{G}=(V,E,\mt H)$ be a variation graph, and let $\mt P$ be a
saturated pebbling of $\mathcal{G}$. In addition to the representation
of the underlying DAG, there exists a data structure using
$O(|E|+|\mt P|\log|\mt H|)$ bits that supports an edge query for any
$e\in E$ as follows:
\begin{itemize}
    \item if $e$ is pebbled, then $\col(e)$ can be reported in
    $O(\log\log |V|+|\col(e)|)$ time;
    \item otherwise, $\col(e)$ can be reported in
    $\widetilde{O}(|\kappa(e)|+|\partial(e)|+|\mt H|)$ time,
\end{itemize}
where $\kappa(e)$ and $\partial(e)$ are defined with respect to the
pebbled edges of $\mt P$. In the latter case, the query requires $O\bigl((|\kappa(e)|+|\partial(e)|+|\mt H|)\log |V|\bigr)$
bits of additional working space.
\end{theorem}

\section{Future Work}
\label{sect-conclusion}

Several promising directions remain for further investigation.
While the present framework assumes that the underlying graph is a DAG, many pangenome graphs contain cycles or more complex bidirectional structures. Extending our problem definitions to general graphs is therefore a natural direction for future work.
Another important extension is to generalize edge queries on subpaths. Specifically, given an arbitrary subpath of the graph, the goal is to identify the set of predefined paths that traverse it. Notably, this functionality is naturally supported by the GBWT.
It would also be interesting to investigate whether tractable instances of \MINWP{} exist, despite its hardness in the general case.
Finally, one could explore alternative criteria for defining guardians beyond the unique path condition, e.g., unique \emph{minimum} paths 
or, more generally, unique paths satisfying a given property, opening up the possibility to encode (in the model) specific knowledge
about the structure of the input graphs.

\section*{Acknowledgements}
The authors are thankful to Giovanni Buzzega, Giulia Punzi, and Nadia Pisanti for fruitful discussions during the early stages of this project.
P.B., G.D.V., Y.G. and B.R., have received funding from the grant MIUR 2022YRB97K, PINC, Pangenome Informatics: from Theory to Applications, funded by the EU, Next-Generation EU, Mission 4.
P.B., G.D.V. and Y.G. are also supported by Commissione Europea - Fostering Excellence in Advanced Genomics and Proteomics Research at Comenius University in Bratislava - FORGENOM II.

%%
%% Bibliography
%%

%% Please use bibtex, 

\bibliography{colored-pebbling-main}

@InProceedings{Milani2024,
author="Milani, Marcelo Garlet",
editor="Soto, Jos{\'e} A.
and Wiese, Andreas",
title="Directed Ear Anonymity",
booktitle="LATIN 2024: Theoretical Informatics",
year="2024",
publisher="Springer Nature Switzerland",
address="Cham",
pages="77--97",
isbn="978-3-031-55601-2"
}

@InProceedings{Dentiet-al2026,
author="Denys Andrukhovskyi and Martin Madzin and Luca Denti and Tomas and Brona",
editor="Nadia El-Mabrouk and Fabio Vandin",
title="Efficient Algorithms for Pangenome Personalization",
booktitle="WABI, 2026",
year="2026",
publisher="Leibniz International Proceedings in Informatics",
pages="9--16",
}

@article{baetz2014brooks,
  title={Brooks' vertex-colouring theorem in linear time},
  author={Baetz, Bradley and Wood, David R},
  journal={arXiv preprint arXiv:1401.8023},
  year={2014}
}

@inproceedings{brooks1941colouring,
  title={On colouring the nodes of a network},
  author={Brooks, Rowland Leonard},
  booktitle={Mathematical Proceedings of the Cambridge Philosophical Society},
  volume={37},
  number={2},
  pages={194--197},
  year={1941},
  organization={Cambridge University Press}
}

@inproceedings{miRCCW98,
  author       = {Venkatesan Guruswami and
                  C. Pandu Rangan and
                  Maw{-}Shang Chang and
                  Gerard J. Chang and
                  C. K. Wong},
  editor       = {Juraj Hromkovic and
                  Ondrej S{\'{y}}kora},
  title        = {The Vertex-Disjoint Triangles Problem},
  booktitle    = {Graph-Theoretic Concepts in Computer Science, 24th International Workshop,
                  {WG} '98, Smolenice Castle, Slovak Republic, June 18-20, 1998, Proceedings},
  series       = {Lecture Notes in Computer Science},
  volume       = {1517},
  pages        = {26--37},
  publisher    = {Springer},
  year         = {1998},
  url          = {https://doi.org/10.1007/10692760\_3},
  doi          = {10.1007/10692760\_3},
  bibsource    = {dblp computer science bibliography, https://dblp.org}
}

@inproceedings{caceres2022sparsifying,
  title={Sparsifying, shrinking and splicing for minimum path cover in parameterized linear time},
  author={C{\'a}ceres, Manuel and Cairo, Massimo and Mumey, Brendan and Rizzi, Romeo and Tomescu, Alexandru I},
  booktitle={Proceedings of the 2022 annual ACM-SIAM symposium on discrete algorithms (SODA)},
  pages={359--376},
  year={2022},
  organization={SIAM}
}

@book{navarro2016compact,
  title={Compact data structures: A practical approach},
  author={Navarro, Gonzalo},
  year={2016},
  publisher={Cambridge University Press}
}

@techreport{BurrowsWheeler1994,
  author      = {Burrows, Michael and Wheeler, David J.},
  title       = {A Block-sorting Lossless Data Compression Algorithm},
  institution = {Digital Equipment Corporation},
  year        = {1994},
  number      = {SRC-TR-124},
  address     = {Palo Alto, CA, USA},
  month       = {May},
}

@inproceedings{fm-index-2000,
author = {Ferragina, P. and Manzini, G.},
title = {Opportunistic data structures with applications},
year = {2000},
isbn = {0769508502},
publisher = {IEEE Computer Society},
address = {USA},
booktitle = {Proceedings of the 41st Annual Symposium on Foundations of Computer Science},
pages = {390},
series = {FOCS '00}
}

@article{alanko2023themisto,
  title={Themisto: a scalable colored k-mer index for sensitive pseudoalignment against hundreds of thousands of bacterial genomes},
  author={Alanko, Jarno N and Vuohtoniemi, Jaakko and M{\"a}klin, Tommi and Puglisi, Simon J},
  journal={Bioinformatics},
  volume={39},
  number={Supplement\_1},
  pages={i260--i269},
  year={2023},
  publisher={Oxford University Press}
}

@article{siren2014indexing,
  title={Indexing graphs for path queries with applications in genome research},
  author={Sir{\'e}n, Jouni and V{\"a}lim{\"a}ki, Niko and M{\"a}kinen, Veli},
  journal={IEEE/ACM transactions on computational biology and bioinformatics},
  volume={11},
  number={2},
  pages={375--388},
  year={2014},
  publisher={IEEE}
}

@article{equi2023complexity,
  title={On the complexity of string matching for graphs},
  author={Equi, Massimo and M{\"a}kinen, Veli and Tomescu, Alexandru I and Grossi, Roberto},
  journal={ACM Transactions on Algorithms},
  volume={19},
  number={3},
  pages={1--25},
  year={2023},
  publisher={ACM New York, NY}
}

@article{baaijens2022computational,
  title={Computational graph pangenomics: a tutorial on data structures and their applications},
  author={Baaijens, Jasmijn A and Bonizzoni, Paola and Boucher, Christina and Della Vedova, Gianluca and Pirola, Yuri and Rizzi, Raffaella and Sir{\'e}n, Jouni},
  journal={Natural computing},
  volume={21},
  number={1},
  pages={81--108},
  year={2022},
  publisher={Springer}
}

@article{computational2018computational,
    author = {The Computational Pan-Genomics Consortium },
    title = {Computational pan-genomics: status, promises and challenges},
    journal = {Briefings in Bioinformatics},
    volume = {19},
    number = {1},
    pages = {118-135},
    year = {2016},
    month = {10},
    issn = {1477-4054},
}

@article{fan2024fulgor,
  title={Fulgor: a fast and compact k-mer index for large-scale matching and color queries},
  author={Fan, Jason and Khan, Jamshed and Singh, Noor Pratap and Pibiri, Giulio Ermanno and Patro, Rob},
  journal={Algorithms for Molecular Biology},
  volume={19},
  number={1},
  pages={3},
  year={2024},
  publisher={Springer}
}

@inproceedings{golynski2006rank,
  author       = {Alexander Golynski and
                  J. Ian Munro and
                  S. Srinivasa Rao},
  title        = {Rank/select operations on large alphabets: a tool for text indexing},
  booktitle    = {Proceedings of the Seventeenth Annual {ACM-SIAM} Symposium on Discrete
                  Algorithms, {SODA} 2006, Miami, Florida, USA, January 22-26, 2006},
  pages        = {368--373},
  publisher    = {{ACM} Press},
  year         = {2006},

}

@article{cairo2022safety,
  title={Safety in s-t Paths, Trails and Walks},
  author={Cairo, Massimo and Khan, Shahbaz and Rizzi, Romeo and Schmidt, Sebastian and Tomescu, Alexandru I},
  journal={Algorithmica},
  volume={84},
  number={3},
  pages={719--741},
  year={2022},
  publisher={Springer}
}

@article{siren2020haplotype,
  title={Haplotype-aware graph indexes},
  author={Sir{\'e}n, Jouni and Garrison, Erik and Novak, Adam M and Paten, Benedict and Durbin, Richard},
  journal={Bioinformatics},
  volume={36},
  number={2},
  pages={400--407},
  year={2020},
  publisher={Oxford University Press}
}

@inproceedings{siren2017indexing,
  title={Indexing variation graphs},
  author={Sir{\'e}n, Jouni},
  booktitle={2017 Proceedings of the ninteenth workshop on algorithm engineering and experiments (ALENEX)},
  pages={13--27},
  year={2017},
  organization={SIAM}
}

@article{liao2023draft,
  title={A draft human pangenome reference},
  author={Liao, Wen-Wei and Asri, Mobin and Ebler, Jana and Doerr, Daniel and Haukness, Marina and Hickey, Glenn and Lu, Shuangjia and Lucas, Julian K and Monlong, Jean and Abel, Haley J and others},
  journal={Nature},
  volume={617},
  number={7960},
  pages={312--324},
  year={2023},
  publisher={Nature Publishing Group UK London}
}

@inproceedings{onodera2013detecting,
  title={Detecting superbubbles in assembly graphs},
  author={Onodera, Taku and Sadakane, Kunihiko and Shibuya, Tetsuo},
  booktitle={International workshop on algorithms in bioinformatics},
  pages={338--348},
  year={2013},
  organization={Springer}
}

\end{document}